\documentclass[twocolumn,final,10pt,journal,twoside]{IEEEtran}

\usepackage{ifthen}
\usepackage{tabu}
\usepackage{amsmath,amsthm,graphicx,cite}
\usepackage{srcltx}
\usepackage{epsfig,amsfonts,subfigure}
\usepackage{graphicx,cite,amssymb,amsmath}
\usepackage{color}

\usepackage{amsthm}
\usepackage[nolist]{acronym}
\usepackage{psfrag}
\usepackage{perso}
\usepackage{enumitem}
\usepackage{bbold}

\usepackage{pgfplots}
 \pgfplotsset{compat=newest}
    \pgfplotsset{plot coordinates/math parser=false}
    \pgfplotsset{
    label style={anchor=near ticklabel},
    xlabel style={yshift=0.0em},
    ylabel style={yshift=-0.3em},
    tick label style={font=\footnotesize },
    label style={font=\footnotesize},
    legend style={font=\footnotesize},
    title style={font=\fontsize{7}}}
\usepackage{xcolor}
\definecolor{iso}{rgb}{0.7,0.7,0.7}
\usepackage{blindtext}
\usepackage{relsize}
\usepackage{mathtools}
\usepackage{bm}
\mathtoolsset{showonlyrefs}

\usepackage[section]{placeins}

\usepgflibrary{arrows}
\usetikzlibrary{patterns}
\usepgflibrary{decorations.pathmorphing}
\usepackage{soul}

\newcommand{\codensemble}{\msr{C}}

\newcommand{\rosymb}{y} 

\newcommand{\absoverhead}{\delta}

\newcommand{\Pf } { \mathsf{P}_{\mathsf{F}}}

\newcommand{\barPf } { \bar {\mathsf{P}}_{\mathsf{F}}}

\newcommand{\Grx}{\tilde{\mathbf{G}}}

\newcommand{\dmax}{ d_{\max}}

\newcommand{\x}{x}

\newcommand{\z}{z}

\newcommand{\Raptorinput}{u}
\newcommand{\vecu}{\mathbf{\Raptorinput}}

\newcommand{\Rintermsymbol}{v}
\newcommand{\vecv}{{\mathbf{\Rintermsymbol}} \mkern-2mu}

\newcommand{\vecgen}{r}

\newcommand{\Rosymb}{c}

\newcommand{\precodegeneric}{\mathcal{C}}

\newcommand{\GrxLT}{\Grx_{\text{LT}}}
\newcommand{\GLT}{\G_{\text{LT}}}
\newcommand{\GLTA}{\G_{\text{LT}}^A}
\newcommand{\GLTB}{\G_{\text{LT}}^B}

\newcommand{\Gp}{\G_{\text{o}}}

\newcommand{\Rrosymb}{\rosymb}

\newcommand{\we}{A}
\newcommand{\weo}{A}
\newcommand{\weodual}{B}

\newcommand{\weoensemble}{\mathsf{A}}

\newcommand{\pil}{\pi_{\l}}
\newcommand{\pilt}{\pi_{\l,t}}

\renewcommand{\l}{l}

\ifthenelse{\isundefined{\G}}{
    \newcommand{\G}{\mathbf{G}}
    }{
    \renewcommand{\G}{\mathbf{G}}
    }

\newcommand{\Exp}{\mathbb{E}}

\newcommand{\Omegarten}{ {\Omega}_{\mathsf{A}} }

\newtheorem{prop}{Proposition}
\newtheorem{theorem}{Theorem}
\newtheorem{remark}{Remark}
\newtheorem{example}{Example}
\newtheorem{lemma}{Lemma}
\newtheorem{corollary}{Corollary}

\ifthenelse{\isundefined{\A}}{
    \newcommand{\A}{\we}
    }{
    \renewcommand{\A}{\we}
    }

\ifthenelse{\isundefined{\argmax}}{
\newcommand{\argmax}{{\arg\max}}
}{
\renewcommand{\argmax}{{\arg\max}}
}

\newcommand{\krawt}{ \mathcal{K}}

\newcommand{\krawtchouk}[4]{\krawt_{#1}^{#3,#4}(#2)}

\newcommand{\rank}{\mathsf{rank}}

\newcommand{\pifroml}{\vartheta_{i,l,j} }
\newcommand{\pifromlA}{\vartheta_{i,l,j}^{(A)} }
\newcommand{\pifromlB}{\vartheta_{d,t,s}^{(B)} }

\newcommand{\qi}{\varphi_i }
\newcommand{\qiAB}{\varphi_{i,d} }

\newcommand{\lab}{f}
\newcommand{\labv}{\mathbf{\lab}\mkern-2mu}

\newcommand{\labg}{\gamma\mkern-2mu}
\newcommand{\labgv}{\bm{\labg}}
\newcommand{\Labgv}{\mathbf{\Gamma}\mkern-2mu}
\newcommand{\Labgvsetj}{\mathsf{\Gamma}_j\mkern-2mu}
\newcommand{\Labgvsets}{\mathsf{\Gamma}_s\mkern-2mu}

\newcommand{\lenum}{\mathcal{Q}_\labv}
\newcommand{\lenumbi}{\mathcal{Q}_{\labv_A, \labv_B}}

\newcommand{\lenumensemble}{\mathsf{Q}_\labv}
\newcommand{\lenumbiensemble}{\mathsf{Q}_{\labv_A, \labv_B}}

\newcommand{\wh}{{w}}
\newcommand{\hw}{\wh}

\newcommand{\comp}{{\varsigma}}
\newcommand{\charact}{{\chi}}

\newcommand{\JWEF}{J}
\newcommand{\JCEF}{\mathcal{S}}
\newcommand{\avgJWEF}{\mathsf{J}}
\newcommand{\avgJCEF}{\mathsf{S}}
\newcommand{\cD}{\mathscr{D}}
\newcommand{\cT}{\mathscr{T}}
\newcommand{\jw}{\tau}
\newcommand{\jc}{\kappa}
\newcommand{\bkappa}{\bm{\kappa}}
\newcommand{\ensemble}{\mathscr{C}}
\newcommand{\btau}{\bm{\tau}}

\newcommand{\cK}{\mathscr{K}}

\newcommand{\bgamma}{\bm{\gamma}}
\newcommand{\bupsilon}{\bm{\upsilon}}
\newcommand{\bUpsilon}{\bm{\Upsilon}}

\newcommand{\bfz}{\mathbf{z}}
\newcommand{\transpose}[1]{#1^{\mathsf{T}}}

\newcommand{\GF}[1]{\mathbb{F}_{\!#1}}
\newcommand{\bbF}[1]{\GF{#1}}

\makeatletter
\newcommand{\vast}{\bBigg@{4}}
\newcommand{\Vast}{\bBigg@{5}}
\makeatother

\begin{document}
\begin{acronym}
\acro{BEC}{binary erasure channel}
\acro{DFT}{discrete Fourier transform}
\acro{$q$-EC}{$q$-ary erasure channel}
\acro{WE}{weight enumerator}
\acro{WEF}{weight enumerator function}
\acro{IOWEF}{input output weight enumerator function}
\acro{IOWE}{input output weight enumerator}
\acro{LT}{Luby Transform}
\acro{BP}{belief propagation}
\acro{ML}{maximum likelihood}
\acro{MDS}{maximum distance separable}
\acro{LDPC}{low-density parity-check}
\acro{i.i.d.}{independent and identically distributed}
\acro{VN}{variable node}
\acro{CN}{check node}
\end{acronym}

\title{Bounds on the Error Probability of Raptor Codes under Maximum Likelihood Decoding}

\author{
	Francisco L\'azaro, \IEEEmembership{Member, IEEE}, Gianluigi Liva, \IEEEmembership{Senior Member, IEEE},\\
	Gerhard Bauch, \IEEEmembership{Fellow, IEEE},
	Enrico Paolini, \IEEEmembership{Senior Member, IEEE} 
	\thanks{Francisco L\'azaro and Gianluigi Liva are with the Institute of Communications and
		Navigation of the German Aerospace Center (DLR), Muenchner Strasse 20, 82234
		Wessling, Germany.
		Email:\{\texttt{Francisco.LazaroBlasco}, \texttt{Gianluigi.Liva}\}\texttt{@dlr.de}.}
	\thanks{Gerhard Bauch is with the Institute for Telecommunication,  Hamburg University of Technology, Hamburg, Germany.
		E-mail: \texttt{Bauch@tuhh.de}.}
    \thanks{Enrico Paolini is with CNIT, DEI, University of Bologna, via Dell'Universit{\`a} 50, 47522 Cesena (FC), Italy.
        E-mail: \texttt{e.paolini@unibo.it}.}	
	\thanks{Corresponding Address:
		Francisco L\'azaro, KN-SAN, DLR, Muenchner Strasse 20, 82234 Wessling, Germany. Tel: +49-8153 28-3211, Fax: +49-8153 28-2844, E-mail: \texttt{Francisco.LazaroBlasco@dlr.de}.}
    \thanks{This work has been presented in part at IEEE Globecom, Washington DC, USA, December 2016 \cite{lazaro:Globecom2016}.
	}
	\thanks{Copyright $\copyright$ 2020 IEEE. Personal use of this material is permitted. However, permission to use this material for any other purposes must be obtained from the IEEE by sending a request to pubs-permissions@ieee.org.}
	\thanks{Digital Object Identifier: 10.1109/TIT.2020.3049061}
}
\maketitle

\thispagestyle{empty}
\pagestyle{empty}

\begin{abstract}
In this paper upper and lower bounds on the probability of decoding failure under \acl{ML} decoding are derived for different (nonbinary) Raptor code constructions.
 In particular four different constructions are considered; \emph{(i)} the standard Raptor code construction, \emph{(ii)} a multi-edge type construction, \emph{(iii)} a construction where the Raptor code is nonbinary but the generator matrix of the LT code has only binary entries, \emph{(iv)} a combination of (ii) and (iii). The latter construction resembles the one employed by RaptorQ codes, which at the time of writing this article represents the state of the art in fountain codes. 
The bounds are shown to be tight, and provide an important aid for the design of Raptor codes.
\end{abstract}


\section{Introduction}\label{sec:Intro}

\IEEEPARstart{F}{ountain} codes \cite{byers02:fountain} are a class of erasure codes which have the property of being rateless. Thus, they are potentially able to generate an endless amount of encoded (or output) symbols from $k$ information (or input) symbols.  This property makes them suitable for application in situations where the channel erasure rate is not a priori known.
The first class of practical fountain codes, \ac{LT} codes, was introduced in \cite{luby02:LT} together with an iterative  decoding algorithm that achieves a good performance when the number of input symbols  is large. In \cite{luby02:LT,shokrollahi06:raptor}  it was shown how, in order to achieve a low probability of decoding error, the encoding and iterative decoding cost\footnote{In \cite{shokrollahi06:raptor} the cost per output symbol is defined as the encoding/decoding complexity normalized by the number of output symbols. The complexity is defined as the number operations needed to carry out encoding/decoding.} per output symbol is $O \left(\ln(k)\right)$.

Raptor codes were introduced in \cite{shokrollahi06:raptor} and outperform \ac{LT} codes in several aspects. They consist of a serial concatenation of an outer code $\mathcal C$ (or \emph{precode}) with an inner \ac{LT} code. On erasure channels, this construction allows relaxing the design of the \ac{LT} code, requiring only the recovery of a fraction $1-\sigma$ of the input symbols, with $\sigma$ small. This can be achieved with linear encoding and decoding complexity  (under iterative decoding). The outer code is responsible for recovering the remaining fraction $\sigma$ of input symbols. If the outer code $\mathcal C$ is linear-time encodable and decodable, then the Raptor code has linear encoding and (iterative) decoding complexity over erasure channels.

Most of the existing works on \ac{LT} and Raptor codes consider iterative decoding and assume large input block lengths ($k$ at least in the order of a few tens of thousands). However, in practice, smaller values of $k$ are more commonly used. For example, for the binary Raptor codes standardized in \cite{MBMS16:raptor} and \cite{luby2007rfc} the supported values of $k$ range from $4$ to $8192$.  For these input block lengths, iterative decoding performance degrades considerably. In this regime,
a different decoding algorithm may be adopted that is an efficient \ac{ML} decoder, in the form of inactivation decoding \cite{berlekamp1968algebraic,lamacchia91:solving,fekri:ldpc,miller04:bec,shokrollahi2005systems}.
An inactivation decoder solves a system of equations in several stages. First a set of variables is declared to be \emph{inactive}. Next a system of equations involving only the set of inactive variables needs to be solved, for example using Gaussian elimination. Finally, once the value of the inactive variables is known, all other variables (those which were not inactive) are recovered using iterative  decoding (back substitution).

Recently, some works have addressed the complexity of inactivation decoding for Raptor and \ac{LT} codes \cite{mahdaviani2012raptor,lazaro:ITW,lazaro:scc2015,lazaro2017inactivation}.
The probability of decoding failure of \ac{LT} and Raptor codes under \ac{ML} decoding has also been subject of study in several works. In \cite{Rahnavard:07}  upper and lower bounds on the  symbol erasure rate were derived for \ac{LT} codes and Raptor codes with outer codes in which the elements of the parity-check matrix are \ac{i.i.d.} Bernoulli random variables.
This work was elegantly extended in \cite{schotsch:2013,Schotsch:14}, where upper and lower bounds on the error probability of \ac{LT} codes under \ac{ML} decoding were derived.  Moreover, \cite{Schotsch:14} introduced an approximation to the probability of error of Raptor codes under \ac{ML} decoding, that was derived under the assumption that the number of erasures correctable by the outer code is small. Hence, the approximation holds when the rate of the outer code is sufficiently high.
In \cite{Liva10:fountain} it was shown by means of simulations how the error probability 
of Raptor codes constructed on $\GF{q}$, the finite field of order $q$,
 is very close to that of linear random fountain codes.
In \cite{wang:2015} upper and lower bounds on the probability of decoding failure of Raptor codes were derived. The outer codes considered in \cite{wang:2015} are binary linear random codes with a systematic encoder. Ensembles of Raptor codes with linear random outer codes were also studied in a fixed-rate setting in \cite{lazaro:ISIT2015},\cite{lazaro:JSAC}.
In \cite{Zhang:bounds}, $q$-ary Raptor codes are considered, but only for the case in which the outer code is a low-density generator matrix code.
Although a number of works have studied the probability of decoding failure of Raptor codes, to the best of the authors' knowledge, up to now the results hold only for specific outer codes (see \cite{Rahnavard:07,wang:2015,lazaro:ISIT2015,lazaro:JSAC,Zhang:bounds}).

In this paper upper and lower bounds on the probability of decoding failure of different Raptor code constructions are derived. 
The upper bounds derived in this paper follow the footsteps of \cite{schotsch:2013,Schotsch:14}, where bounds to the error probability of LT codes were derived. In contrast to other works in literature \cite{Rahnavard:07,wang:2015,lazaro:ISIT2015,lazaro:JSAC,Zhang:bounds}, the bounds presented in this paper are general since they are valid for any outer code, requiring only the (joint) weight enumerator (or composition enumerator, a quantity to be defined later) of the outer code.  Furthermore, simulation results are presented which show how the derived bounds are tight.
In particular four different constructions are considered, namely:
\begin{enumerate}[label=\roman*)]
  \item  a Raptor code construction over $\GF{q}$, where the outer code is built over $\GF{q}$ as well as the generator matrix of the LT code;
  \item a multi-edge type Raptor construction over $\GF{q}$, where intermediate symbols of two different types can be distinguished;
  \item a construction where the Raptor code is built over $\GF{q}$ but the generator matrix of the LT code has only entries belonging to $\{0,1\} \subseteq \GF{q}$;
  \item a combination of (ii) and (iii).
\end{enumerate}
The bounds are applicable for the two Raptor codes present in standards. In particular, the R10 Raptor code in its nonsystematic form \cite{MBMS16:raptor} is an example of construction (i), since binary Raptor codes are simply a special case ($q=2$). Furthermore, the RaptorQ code in its nonsystematic form \cite{lubyraptorq} is an example of construction (iv).
The RaptorQ code is,  at the timing, the state of the art fountain code construction, and it is an {IETF} standard \cite{lubyraptorq}. To the best of the authors' knowledge, this is the first work which analyzes the performance of the RaptorQ construction\footnote{In \cite{Zhang:bounds} a $q$-ary Raptor code construction is analyzed, but it does not consider all the peculiarities of the RaptorQ code.}.

The upper bounds on the probability of decoding failure  are derived for all the above four constructions and they all result from application of the union bound. As mentioned before, they generalize the results in literature to the case where the outer codes are chosen arbitrarily (with the caveat of having sufficient knowledge of the outer code distance properties). In the same general setting, two types of lower bounds are obtained. A first lower bound is a consequence of the degree-two Bonferroni inequality (as for the lower bounds introduced in \cite{Rahnavard:07}). A second, tighter lower bound is obtained by means of the Dawson-Sankoff inequality \cite{dawson67:inequality}, which generalizes the Bonferroni inequality.\footnote{Note that the Dawson-Sankoff inequality was used in \cite{Barak07} to lower bound the expected error probability of regular \ac{LDPC} code ensembles over the \ac{BEC}.} The bounds are shown to be remarkably tight at large overheads, and sufficiently tight at overheads approaching zero. Starting from the upper bound on the probability of decoding failure, an error exponent analysis of Raptor codes is presented, which allows characterizing the overhead regions for which an exponential decay (in the input block length) of the expected failure probability can be attained.
Examples of the application of the proposed bounds to the design of Raptor codes are finally provided.

The paper is organized as follows. In Section~\ref{sec:prelim} some preliminary definitions are given. Section~\ref{sec:results_joint} presents a  number of results on joint compositions.
Section~\ref{sec:constructions} addresses the different Raptor code constructions considered in this paper. Section~\ref{sec:bounds} presents several theorems with upper and lower bounds on the probability of decoding failure for the different Raptor code constructions. Proofs of the bounds are given in Section~\ref{sec:proofs}. Section \ref{sec:errexp} introduces the error exponent analysis. Numerical results comparing the bounds with Monte Carlo simulations are illustrated in Section~\ref{sec:numres}, while code design examples are discussed in Section \ref{sec:code_design}. Section~\ref{sec:Conclusions} presents the conclusions of our work.

\section{Preliminaries}\label{sec:prelim}

\subsection{Vector and Matrix Notation}

We use boldface letters to denote vectors and matrices. Vectors are conventionally assumed as row vectors with indices starting from~$0$; matrix row and column indices also start from~$0$. For any integer matrix $\mathbf{A}$ we denote by $|\mathbf{A}|$ the sum of all matrix elements. We use the same notation for integer vectors, i.e., $|\mathbf{a}|$ represents the sum of all elements of vector $\mathbf{a}$. We also denote by $\mathbf{1}(\mathbf{A})$ the matrix obtained from $\mathbf{A}$ by turning to $1$ all its nonzero elements. The transpose of any matrix $\mathbf{A}$ is denoted by $\transpose{\mathbf{A}}$.

We say that a zero-one square matrix $\mathbf{A}$ is a circulant permutation matrix when: (i) it is a permutation matrix; (ii) each row of $\mathbf{A}$ is obtained from the previous row by the right cyclic shift of one position. We say that a zero-one square matrix $\mathbf{A}$ is an \emph{incomplete} circulant permutation matrix when: (i) it is nonzero; (ii) it can be obtained from a circulant permutation matrix by turning to $0$ some $1$ elements.

For a nonnegative integer vector $\mathbf{a}=(a_0,a_1,\dots,a_{n-1})$ such that $|\mathbf{a}|=h$ we denote the multinomial coefficient ${\binom{h}{a_0, a_1, \dots,a_{n-1}}}$ by ${\binom{h} {\mathbf{a}}}$. With a slight abuse of notation, for an $m \times n$ nonnegative integer matrix $\mathbf{A}=[a_{s,t}]$ such that $|\mathbf{A}|=h$ we write ${\binom{h}{\mathbf{A}}}$ as a compact notation for 
${\binom{h}{a_{0,0}, \dots, a_{0,n-1}, \dots, a_{m-1,0}, \dots, a_{m-1,n-1}}}$.

\subsection{Bonferroni-Type Inequalities}

Let $A_{1}$, $\dots$, $A_{n}$ be events in a probability space and
\begin{align*}
S_k = \sum_{1 \leq i_1 < \dots < i_k \leq n} \Pr \{A_{i_1} \cap \dots \cap A_{i_k} \} .
\end{align*}
The general Bonferroni inequality states that, for any $1 \leq t \leq n$, we have \cite{Bonferroni36}
\begin{align}\label{eq:bonferroni_general}
(-1)^t\, \Pr \{ A_1 \cup \dots \cup A_{n}\} \geq (-1)^t \sum_{i=1}^t (-1)^{i-1} S_i .
\end{align}
Inequality~\eqref{eq:bonferroni_general} holds with equality for $t=n$ (inclusion-exclusion identity). Notable special cases are obtained for $t=1$ and $t=2$. Specifically, for $t=1$ it reduces to the union upper bound
\begin{align}\label{eq:boole}
\Pr \{ A_1 \cup \dots \cup A_{n} \} \leq S_1 = \sum_{i=1}^n \Pr \{A_i\}
\end{align}
while for $t=2$ it yields the degree-two Bonferroni lower bound
\begin{align}\label{eq:order2_bonferroni}
\Pr \{ A_1 \cup \dots \cup A_{n} \} &\geq S_1 - S_2 \notag \\ 
& = \sum_{i=1}^n \Pr \{A_i\} - \sum_{1 \leq i < j \leq n} \Pr \{ A_i \cap A_j \} .
\end{align}
A tighter version of \eqref{eq:order2_bonferroni} was developed in \cite{dawson67:inequality}, where it was shown that, for any $r \in \{1,\dots,n\}$,
\begin{align}\label{eq:pre_dawson_sankoff}
\Pr \{ A_1 \cup \dots \cup A_{n} \} \geq \frac{2}{r+1} S_1 - \frac{2}{r(r+1)} S_2 .
\end{align}
Moreover, maximization with respect to $r$ yields
\begin{align}\label{eq:dawson_sankoff}
\Pr \{ A_1 \cup \dots \cup A_{n} \} \geq \frac{\theta S_1^2}{(2-\theta)S_1 + 2 S_2} + \frac{(1-\theta)S_1^2}{(1-\theta)S_1+2S_2}
\end{align}
where $\theta = 2 S_2 / S_1 - \lfloor 2 S_2 / S_1 \rfloor$. Indeed, it was proved in \cite{kwerel75:most} that \eqref{eq:dawson_sankoff} is the sharpest possible lower bound for ${\Pr\{A_1 \cup \cdots \cup A_n\}}$ based on a linear combination of $S_1$ and $S_2$. As such, it is tighter than $S_1-S_2$. Hereafter, \eqref{eq:dawson_sankoff} will be referred to as Dawson-Sankoff lower bound.

\subsection{Weight and Composition Enumerators}

For any linear block code $\precodegeneric$ constructed over $\mathbb{F}_q$ and any codeword $\vecv \in \mathcal{C}$, we let $\hw(\vecv)$ be the Hamming weight (often referred to simply as the weight) of $\vecv$. Letting $h$ be the codeword length, we  denote the weight enumerator of $\mathcal{C}$
as $\weo = \{\weo_0, \weo_1 \hdots \weo_h\}$, where $\weo_i$ denotes the multiplicity of codewords of weight $i$.
Similarly, given an ensemble $\codensemble$ of linear block codes, all with the same block length $h$, along with a probability distribution on the codes in the ensemble, we denote the expected weight enumerator of a random code in $\codensemble$ as ${\weoensemble = \{\weoensemble_0, \weoensemble_1 \hdots \weoensemble_h\}}$, where $\weoensemble_l$ denotes the expected multiplicity of codewords of weight $l$.

Next, consider a linear block code $\precodegeneric \subset \mathbb{F}_q^h$, whose codeword symbols are partitioned into two different types, namely, type $A$ and type $B$. Let $h_A$ and $h_B$ be the number of codeword symbols of types $A$ and $B$, respectively, such that ${h_A+h_B=h}$. A generic codeword after reordering can be expressed as
${\vecv=(\vecv_A,\vecv_B)}$, where $\vecv_A$ and $\vecv_B$ denote the vectors of encoded symbols of type $A$ and type $B$ respectively. In this context the bivariate weight enumerator polynomial of the code is defined as
\begin{equation}\label{eq:weo_bivariate}
\weo(x,z) = \sum_{l=0}^{h_A} \sum_{t=0}^{h_B}  \weo_{l,t} \, x^l z^t
\end{equation}
where $\weo_{l,t}$ denotes the multiplicity of codewords with $\hw(\vecv_A)=l$  and $\hw(\vecv_B)=t$.
Similarly, given an ensemble $\codensemble$ of block codes  with block length $h$ and with two types of codeword symbols as defined above, along with a probability distribution on the codes in the ensemble, we define its expected bivariate weight enumerator polynomial as
\[
\weoensemble(x,z) = \sum_{l=0}^{h_A} \sum_{t=0}^{h_B}  \weoensemble_{l,t} \, x^l z^t
\]
where $\weoensemble_{l,t}$ denotes the expected multiplicity of codewords with $\hw(\vecv_A)=l$  and $\hw(\vecv_B)=t$.

Given a vector $\mathbf{\vecgen}=(\vecgen_0, \vecgen_1, \dots, \vecgen_{h-1}) \in \mathbb{F}_q^h$, we define its composition $\comp(\mathbf{\vecgen})$ as
\[
\comp(\mathbf{\vecgen}) = \left( \comp_0(\mathbf{\vecgen}), \comp_1(\mathbf{\vecgen}), \dots,  \comp_{q-1}(\mathbf{\vecgen})  \right)
\]
where
\[
\comp_i(\mathbf{\vecgen}) = \left| \left\{ \vecgen_j:\vecgen_j = \alpha^{i-1} \right\} \right|\]
for $j \in \{0, \dots, h-1 \}$, and $i \in \{1, 2, \dots, q-1\}$,
being $\alpha$ the residue class of the polynomial $x$, and
\[
\comp_0(\mathbf{\vecgen}) = \left| \left\{ \vecgen_j:\vecgen_j =0 \right\} \right| \, \,\text{for } j \in \{1, \dots, h \}.
\]
That is, $\comp_i(\mathbf{\vecgen})$, $i \in \{1, 2, \dots, q-1\}$, is the number of elements in $\mathbf{\vecgen}$ that take value $\alpha^{i-1}$ whereas  $\comp_0(\mathbf{\vecgen})$ is the number of null elements in $\mathbf{\vecgen}$.
Given a linear block code $\precodegeneric$, we define its composition enumerator, $\lenum$,
as the number of codewords $\vecv \in \precodegeneric$ with composition $\comp(\vecv)=\labv$. Similarly, for a code ensemble we define its expected composition enumerator $\lenumensemble$ as the expected multiplicity of codewords with composition~$\labv$.

Consider also a linear block code $\precodegeneric$ of length $h$, with two types of codeword symbols as defined above. We define the bivariate composition enumerator $\lenumbi$ of a code $\precodegeneric$
as the number of codewords  $\vecv= (\vecv_A, \vecv_B)$ in $\precodegeneric$ for which $\vecv_A$ has composition $\labv_A$ and $\vecv_B$ has composition $\labv_B$. This definition can be easily extended to code ensembles. In particular, we define the
expected bivariate composition enumerator $\lenumbiensemble$ of a random code in the ensemble as the expected multiplicity of codewords $\vecv=(\vecv_A, \vecv_B)$ for which $\vecv_A$ has composition $\labv_A$ and $\vecv_B$ has composition $\labv_B$.

Given the composition $\labv$ of a vector $\mathbf{\vecgen} \in \mathbb{F}_q^h$, ${\labv=\comp(\mathbf{\vecgen})}$, as defined above, we define $B(\labv \,)$ as an indicator function that takes value $1$ only if $\sum_{i=1}^{h} \vecgen_i =0$, i.e.,
\begin{equation}
B(\labv \,) = \begin{cases}
                  1, & \mbox{if } \sum_{i=1}^{q-1} \sum_{s=1}^{\lab_i} \alpha^{i-1}=0 \\
                  0, & \mbox{otherwise.}
                \end{cases}
\end{equation}

\subsection{Joint Weight and Joint Composition Enumerators}\label{sec:joint_weight_enum}

Given two vectors $\mathbf{r}_1 \in \bbF{q}^h$ and $\mathbf{r}_2 \in \bbF{q}^h$, we define the joint weight of $\mathbf{r}_1$ and $\mathbf{r}_2$, denoting it as $\btau = \jw(\mathbf{r}_1,\mathbf{r}_2)$, as the vector $(\tau_0,\tau_1,\tau_2,\tau_3)$ such that:
\begin{itemize}
\item There are $\tau_0$ positions in which both $\mathbf{r}_1$ and $\mathbf{r}_2$ are zero;
\item There are $\tau_1$ positions in which $\mathbf{r}_1$ is zero and $\mathbf{r}_2$ is nonzero;
\item There are $\tau_2$ positions in which $\mathbf{r}_1$ is nonzero and $\mathbf{r}_2$ is zero;
\item There are $\tau_3$ positions in which both $\mathbf{r}_1$ and $\mathbf{r}_2$ are nonzero.
\end{itemize}
The elements of $\btau = \jw(\mathbf{r}_1,\mathbf{r}_2)$ are nonnegative integers and $|\btau| = h$.

Given two vectors $\mathbf{r}_1\in \bbF{q}^h$ and $\mathbf{r}_2 \in \bbF{q}^h$, we define the joint composition of $\mathbf{r}_1$ and $\mathbf{r}_2$, denoting it as $\bkappa = \jc(\mathbf{r}_1,\mathbf{r}_2)$, as the $q \times q$ matrix $[\kappa_{s,t}]$, $(s,t) \in \{0,\dots,q-1\} \times \{0,\dots,q-1\}$, such that:
\begin{itemize}
\item There are $\kappa_{0,0}$ positions in which both $\mathbf{r}_1$ and $\mathbf{r}_2$ are zero;
\item There are $\kappa_{0,t}$ positions, $t \neq 0$, in which $\mathbf{r}_1$ is zero and $\mathbf{r}_2$ is equal to $\alpha^{t-1}$;
\item There are $\kappa_{s,0}$ positions, $s \neq 0$, in which $\mathbf{r}_1$ is equal to $\alpha^{s-1}$ and $\mathbf{r}_2$ is zero;
\item There are $\kappa_{s,t}$ positions, $s \neq 0$, $t \neq 0$, in which $\mathbf{r}_1$ is equal to $\alpha^{s-1}$ and $\mathbf{r}_2$ is equal to $\alpha^{t-1}$.
\end{itemize}
The elements of  $\jc(\mathbf{r}_1,\mathbf{r}_2)$ are nonnegative integers and ${|\bkappa| = h}$. We write
\begin{align}\label{eq:bkappa_structure}
\bkappa = \left[ \begin{array}{cc}
\kappa_{0,0} & \bkappa_1 \\
\bkappa_2 & \bkappa_3
\end{array} \right]
\end{align}
where $\bkappa_1$ is the $1 \times (q-1)$ matrix $[\kappa_{0,1}, \dots, \kappa_{0,q-1}]$, $\bkappa_2$ is the $(q-1) \times 1$ matrix $\transpose{[\kappa_{1,0}, \dots, \kappa_{q-1,0}]}$, and $\bkappa_3$ is the $(q-1) \times (q-1)$ matrix $[\kappa_{s,t}]$, $(s,t) \in \{1,\dots,q-1\} \times \{1,\dots,q-1\}$.

There is a simple relationship between the joint weight $\btau=\jw(\mathbf{r}_1,\mathbf{r}_2)$ of two vectors and their joint composition $\bkappa=\jc(\mathbf{r}_1,\mathbf{r}_2)$. In particular, we have $\tau_0 = \kappa_{0,0}$, $\tau_1 = |\bkappa_1|$, $\tau_2 = |\bkappa_2|$, and $\tau_3 = |\bkappa_3|$. We write $\btau = \jw(\bkappa)$ to indicate the joint weight $\btau$ associated with the joint composition $\bkappa$. There also is a simple relationship between the joint composition $\bkappa=\jc(\mathbf{r}_1,\mathbf{r}_2)$ of two vectors and the composition of each of them. Specifically, denoting the composition of $\mathbf{r}_1$, $\varsigma(\mathbf{r}_1)$, by $\bgamma_1(\bkappa)$ and the composition of $\mathbf{r}_2$, $\varsigma(\mathbf{r}_2)$, by $\bgamma_{2}(\bkappa)$, we have
\begin{align}\label{eq:gamma_one}
\bgamma_1(\bkappa) &= \Big(\sum_{t=0}^{q-1} \kappa_{0,t}, \dots, \sum_{t=0}^{q-1} \kappa_{q-1,t} \Big)   
\\
\bgamma_2(\bkappa)&= \Big(\sum_{s=0}^{q-1} \kappa_{s,0}, \dots, \sum_{s=0}^{q-1} \kappa_{s,q-1} \Big) . \label{eq:gamma_two}
\end{align}

Given two linear block codes $\precodegeneric_1 \subset \bbF{q}^h$ of dimension $k_1$ and $\precodegeneric_2 \subset \bbF{q}^h$ of dimension $k_2$, we define their joint weight enumerator, denoting it by $\JWEF_{\btau}$, as the number of codeword pairs $(\vecv,\bfz) \in \precodegeneric_1 \times \precodegeneric_2$ such that $\jw(\vecv,\bfz)=\btau$. We also define their joint composition enumerator, denoting it by $\JCEF_{\bkappa}$, as the number of codeword pairs $(\vecv,\bfz) \in \precodegeneric_1 \times \precodegeneric_2$, such that $\jc(\vecv,\bfz)=\bkappa$. If $\mathcal{C}_1 = \mathcal{C}_2 = \mathcal{C}$, then $\JWEF_{\btau}$ and $\JCEF_{\bkappa}$ are called the biweight and the bicomposition enumerator of $\mathcal{C}$, respectively. For an ensemble $\codensemble$ of linear block codes, all with the same block length, we denote by $\avgJWEF_{\btau}$ and $\avgJCEF_{\bkappa}$ the expected biweight and bicomposition enumerators, respectively, of a random code in~$\codensemble$.\footnote{The concept of joint weight and joint weight enumerator was introduced in \cite{MMS:1972}, where examples of biweight numerators for some classical codes were obtained.}
\begin{remark}\label{remark:binary_equiv}
For $q=2$, if $\btau = \jw(\bkappa)$, then $\btau = (\kappa_{0,0}, \kappa_{0,1}, \kappa_{1,0}, \kappa_{1,1})$. Thus, in the binary case there exists a bijection between joint weights and joint compositions so that the two concepts become equivalent and can be used interchangeably. With this bijection in mind we can also write $\JCEF_{\bkappa} = \JWEF_{\btau}$. This is not the case in the nonbinary case.
\end{remark}

\subsection{Weight Spectral Shape of Code Ensemble Sequences}\label{subsec:growthrate}
A code ensemble sequence $\left\{\codensemble_k\right\}$ is a sequence of code ensembles, where $\codensemble_k$ is an ensemble of dimension-$k$ codes with block length $h=k/R$ defined over $\GF{q}$, being $R$ a constant, i.e., not dependent on $k$. The weight spectral shape of the  
ensemble sequence $\left\{\codensemble_k\right\}$ is given by
\[
G(\omega)=\lim_{h\rightarrow\infty} \frac{1}{h} \log_2 \weoensemble^{(hR)}_{\lfloor \omega h \rfloor}
\]
where $\weoensemble^{(hR)}$ is the expected weight enumerator of the code ensemble $\codensemble_{hR}$. In the definition above, $\omega$ can be regarded as the normalized Hamming weight.

We recall next the definition of uniform convergence, which will become essential for the results derived in Section \ref{sec:errexp}. A sequence $f_h$ of real-valued functions on $D\subseteq \mathbb{R}$ converges uniformly to the function $f: D \mapsto \mathbb{R}$ on $D_0 \subseteq D$ if for any $\varepsilon>0$ there exists $h_0(\varepsilon)$ such that, for all $h\geq h_0(\varepsilon)$, $\left|f_h(x)-f(x)\right|<\varepsilon$ for all $x\in D_0$. We write $f_h \xrightarrow{\mathsf{u}} f $ to indicate that $f_h$ converges to $f$ uniformly.

\subsection{Further Useful Definitions and Results}

For a positive integer $n$ and a prime or prime power $q$, we denote by $\krawtchouk{i}{x}{n}{q}$ the Krawtchouk polynomial of degree $i$ with parameters $n$ and $q$, which is defined as \cite{MacWillimas77:Book}
\[
\krawtchouk{i}{x}{n}{q} = \sum_{j=0}^i (-1)^j \binom{x}{j} \binom{n-x}{i-j} (q-1)^{i-j}.
\]

\noindent Moreover, we recall Chu-Vandermonde identity, stating that
\[
\binom{m+n}{r} = \sum_{k=0}^{r} \binom{m}{k} \binom{n}{r-k}.
\]
\section{Results on Joint Weights and Joint Compositions}\label{sec:results_joint}

This section presents a number of results on joint compositions. These results will be useful to develop a lower bound on the error probability of a class of Raptor codes.
\begin{lemma}\label{lemma:proportional}
Let $\mathbf{r}_1 \in \bbF{q}^h \setminus \{\mathbf{0}\}$ and $\mathbf{r}_2 \in \bbF{q}^h \setminus \{\mathbf{0}\}$. We have
\begin{align*}
\jc(\mathbf{r}_1,\mathbf{r}_2) = \left[ \begin{array}{cc}
\kappa_{0,0} & \mathbf{0} \\
\mathbf{0} & \bkappa_3
\end{array} \right]
\end{align*}
in which $\mathbf{1}(\bkappa_3)$ is a (possibly incomplete) circular permutation matrix, if and only if $\mathbf{r}_1 = \beta \mathbf{r}_2$ for some $\beta \in \GF{q} \setminus \{0\}$.
\end{lemma}
\begin{IEEEproof}
Let $\mathbf{r}_1 = \beta \mathbf{r}_2$ for some $\beta \in \GF{q} \setminus \{0\}$ ($\mathbf{r}_1$ and $\mathbf{r}_2$ are linearly dependent). With reference to \eqref{eq:bkappa_structure}, since $\mathbf{r}_1$ and $\mathbf{r}_2$ have the same support, both $\bkappa_1$ and $\bkappa_2$ must be null. Let the proportionality factor $\beta$ be equal to $\alpha^{s}$ for some $s \in \{0,\dots,q-2\}$. Every element of $\mathbf{r}_1$ equal to $\alpha^i$ corresponds to an element $\alpha^{(i+s)\,\mathrm{mod}\,(q-1)}$ in $\mathbf{r}_2$, making $\kappa_{1+i,1+(i+s)\,\mathrm{mod}\,(q-1)}>0$; any other element of $\bkappa$ in row of index $1+i$ must be zero. This suffices to conclude that $\mathbf{1}(\bkappa_3)$ is a circulant permutation matrix if all elements of $\bbF{q} \setminus \{0\}$ appear in $\mathbf{r}_1$. It is an incomplete circulant permutation matrix otherwise. Conversely, let $\bkappa_1=\bkappa_2=\mathbf{0}$ and $\mathbf{1}(\bkappa_3)$ be a (possibly incomplete) circulant permutation matrix. The vectors $\mathbf{r}_1$ and $\mathbf{r}_2$ must have the same support. Moreover, there must exist $s\in\{0,\dots,q-2\}$ such that every nonzero element of $\bkappa$, apart from $\kappa_{0,0}$, is in the form $\kappa_{1+i,1+(i+s)\,\mathrm{mod}\,(q-1)}$ for some $i \in \{1,\dots,q-1\}$. But then $\mathbf{r}_1 = \alpha^s \mathbf{r}_2$. 
\end{IEEEproof}

Let $\precodegeneric \subset \bbF{q}^h$ be a linear block code of dimension $k$. We partition the codebook of $\precodegeneric$ into $M_{q,k}=(q^k-1)/(q-1)+1$ parts $\mathcal{P}_a$, $a = \{0,1,\dots,M_{q,k}-1\}$, as follows. Part $\mathcal{P}_0$ only contains the null codeword, while any other part contains $q-1$ codewords having the same support and being linearly dependent. Moreover, we index the codewords in $\precodegeneric$ from $0$ to $q^k-1$, as follows. The index $0$ is reserved to the null codeword; the indices from $(a-1)(q-1)+1$ to $a(q-1)$ are reserved to the codewords in $\mathcal{P}_a$, $a \in \{1,\dots,M_{q,k}-1\}$. For every $a \in \{1,\dots,M_{q,k}-1\}$ we take one representative in $\mathcal{P}_a$, denoting it by $\tilde{\vecv}_a$. In particular, we choose as $\tilde{\vecv}_a$ the codeword in $\mathcal{P}_a$ having the smallest index. Letting $\{\mathbf{0}, \vecv_1, \vecv_2, \dots, \vecv_{q^k-1}\}$ be the codebook of $\precodegeneric$, with the above-mentioned indexing convention we have $\tilde{\vecv}_a = \vecv_{(a-1)(q-1)+1}$.

We define the set $\cD_{q,k}$ as 
\begin{align*}
\cD_{q,k} &= \{ (s,t) \in \{1,q^k-1\} \\
& \phantom{=}\times \{1,q^k-1\} : \lfloor \frac{s-1}{q-1} \rfloor \neq \lfloor \frac{t-1}{q-1} \rfloor \} .
\end{align*}
Moreover, we define the set $\tilde{\cD}_{q,k} \subseteq \cD_{q,k}$ as
\begin{align*}
\tilde{\cD}_{q,k} = \{&(s,t) \in \cD_{q,k} : s=(a-1)(q-1)+1; \\
&  t=(b-1)(q-1)+1;   a,b \in \{1,\dots,M_{q,k}-1\}; \\
&a \neq b \} \, .
\end{align*}
The set $\cD_{q,k}$ is the set of codeword index pairs $(s,t)$ such that: (i) $\vecv_s$ and $\vecv_t$ are both nonzero; (ii) $\vecv_s$ and $\vecv_t$ are not linearly dependent. Its cardinality is $(q-1)^2 (M_{q,k}-2) (M_{q,k}-1)$. The set $\tilde{\cD}_{q,k}$ is a subset of $\cD_{q,k}$. It includes all codeword index pairs $(s,t)$ such that $\vecv_s$ is the representative of part $\mathcal{P}_{(s-1)/(q-1)+1}$, $\vecv_t$ is the representative of part $\mathcal{P}_{(t-1)/(q-1)+1}$, and $\vecv_s \neq \vecv_t$. Its cardinality is $(M_{q,k}-2) (M_{q,k}-1)$.
\begin{figure}[!t]
\begin{center}
\includegraphics[width=0.8\columnwidth]{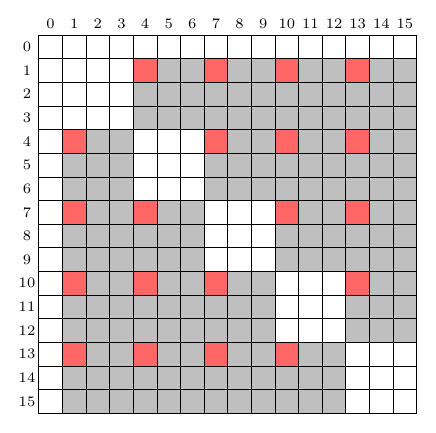}
\end{center}
\caption{Graphical interpretation of the set $\cD_{q,k}$ for $q=4$ and $k=2$.}\label{fig:chessboardq}
\end{figure}
\begin{example}
Let $q=4$ and $k=2$. A graphical interpretation of the set $\cD_{4,2}$ is provided in Fig.~\ref{fig:chessboardq}. The codebook is partitioned into the $M_{4,2}=6$ parts $\mathcal{P}_0=\{\mathbf{0}\}$, $\mathcal{P}_1=\{\vecv_1,\vecv_2,\vecv_3\}$, $\mathcal{P}_2=\{\vecv_4,\vecv_5,\vecv_6\}$, $\mathcal{P}_3=\{\vecv_7,\vecv_8,\vecv_9\}$, $\mathcal{P}_4=\{\vecv_{10},\vecv_{11},\vecv_{12}\}$, $\mathcal{P}_5=\{\vecv_{13},\vecv_{14},\vecv_{15}\}$, where all codewords in the same part are linearly dependent. The set $\cD_{4,2}$ is represented by the union of all grey and red cells of the ``chessboard'', while the set $\tilde{\cD}_{4,2}$ is represented only by the red cells. White cells, the ones not belonging to $\cD_{4,2}$, correspond either to pairs of codewords of which at least one is null or to pairs of linearly dependent codewords.
\end{example}
We define $\cK_{q,h}$ as the set of all joint compositions $\bkappa$ such that $|\bkappa|=h$ and such that any of the following two conditions holds: (1) at least two matrices out of $\bkappa_1$, $\bkappa_2$, $\bkappa_3$ are nonzero; (2) $\bkappa_1$ and $\bkappa_2$ are null matrices, $\bkappa_3$ is nonzero, $\mathbf{1}(\bkappa_3)$ is neither a complete nor an incomplete circulant permutation matrix.
\begin{lemma}\label{lemma:KD}
For any linear block code $\precodegeneric \subset \bbF{q}^h$ of dimension $k$ and any pair $(\vecv_s,\vecv_t) \in \precodegeneric \times \precodegeneric$, we have $\jc(\vecv_s,\vecv_t)\in \cK_{q,h}$ if and only if $(s,t) \in \cD_{q,k}$.
\end{lemma}
\begin{IEEEproof}
Let $\kappa(\vecv_s,\vecv_t) \in \cK_{q,h}$. If at least two matrices out of $\bkappa_1$, $\bkappa_2$, and $\bkappa_3$ are nonzero, then $\vecv_s$ and $\vecv_t$ are both nonzero and have different supports (so they cannot be linearly dependent). Thus we must have $(s,t) \in \cD_{q,k}$. If $\bkappa_1=\bkappa_2=\mathbf{0}$, $\bkappa_3\neq\mathbf{0}$, and $\mathbf{1}(\bkappa_3)$ is neither a circulant permutation matrix nor an incomplete one, then $\vecv_s$ and $\vecv_t$ have the same support but are not linearly dependent (Lemma~\ref{lemma:proportional}). Thus we must have $(s,t) \in \cD_{q,k}$ again. Conversely, let $(s,t) \in \cD_{q,k}$, meaning that $\vecv_s$ and $\vecv_t$ are both nonzero and they are not linearly dependent. If $\vecv_s$ and $\vecv_t$ have different supports then at least two matrices out of $\bkappa_1$, $\bkappa_2$, and $\bkappa_3$ must nonzero, so $\kappa(\vecv_s,\vecv_t) \in \cK_{q,h}$. If $\vecv_s$ and $\vecv_t$ have the same support, since they are not linearly dependent, by Lemma~\ref{lemma:proportional} $\bkappa_3$ can be neither a circulant permutation matrix, nor an incomplete one. Hence $\kappa(\vecv_s,\vecv_t) \in \cK_{q,h}$ again.
\end{IEEEproof}

\subsection{Binary codes}

In Remark~\ref{remark:binary_equiv} we pointed out that over $\bbF{2}$ the concepts of joint composition and joint weight become equivalent. Thus, in the binary case the quantities and results so far introduced in this section can be reformulated in terms of joint weight. Note at first that when $q=2$ the two sets $\cD_{2,k}$ and $\tilde{\cD}_{2,k}$ coincide and that $\cD_{2,k}$ can be simply defined as
\begin{align*}
\cD_{2,k} = \{ (s,t) \in \{1,2^k-1\} \times \{1,2^k-1\} : s \neq t  \} .
\end{align*}
This is the set of all codeword index pairs $(s,t)$ such that $\vecv_s \neq \mathbf{0}$, $\vecv_t \neq \mathbf{0}$, and $\vecv_s \neq \vecv_t$.

For $q=2$, $\cK_{2,h}$ may be simply defined as the set of all joint compositions $\bkappa=[\kappa_{s,t}]$, $s,t \in \{0,1\}$, such that $|\bkappa|=h$ and such that at least two parameters out of $\kappa_{0,1}$, $\kappa_{1,0}$, $\kappa_{1,1}$ are positive. Owing to the above-recalled equivalence between joint weights and joint compositions, we introduce the set $\cT_{2,h}$ as the equivalent of $\cK_{2,h}$ for joint weights. We define $\cT_{2,h}$ as the set of all joint weights $\btau=(\tau_0, \tau_1, \tau_2, \tau_3)$ such that $|\btau|=h$ and such that at least two parameters out of $\tau_1$, $\tau_2$, $\tau_3$ are positive.
The following result is an immediate corollary of Lemma~\ref{lemma:KD} for~$q=2$.
\begin{lemma}\label{lemma:no_tilde}
For any binary linear block code $\precodegeneric \subset \bbF{2}^h$ of dimension $k$ and any pair $(\vecv_s,\vecv_t) \in \precodegeneric \times \precodegeneric$, we have $\jw(\vecv_s,\vecv_t) \in \cT_{2,h}$ if and only if $(s,t) \in \cD_{2,k}$.
\end{lemma}
\section{Raptor Codes}\label{sec:constructions}

\subsection{Encoding and Decoding}

We consider four different Raptor code constructions, all of them  over $\GF{q}$, with $q \geq 2$, being $q$ a prime or prime power.
Fig.~\ref{fig:Raptor_enc} shows a block diagram of Raptor encoding. In particular we consider an outer linear block code $\precodegeneric$  whose length and dimension are denoted by $h$ and $k$, respectively.
We denote the $k$ input (or source) symbols of the Raptor code as ${\vecu=(\Raptorinput_0, \Raptorinput_1, \ldots, \Raptorinput_{k-1})}$.
Out of the $k$ input symbols, the outer code generates a vector of $h$ intermediate symbols ${\vecv=(\Rintermsymbol_0, \Rintermsymbol_1, \ldots, \Rintermsymbol_{h-1})}$. The rate of the outer code is hence $R=k/h$. Denoting by $\Gp$ the generator matrix of the outer code, of dimension $(k \times h)$, the intermediate symbol vector can be expressed as
\[
\vecv = \vecu \Gp.
\]
The intermediate symbols serve as input to an LT encoder, which generates the output symbols ${\mathbf{\Rosymb}=(\Rosymb_0, \Rosymb_1, \ldots, \Rosymb_{n-1})}$, where $n$ can grow unbounded.
For any $n$, we have 
\begin{equation}\label{eq:raptor_enc}
\mathbf{\Rosymb} = \vecv \, \GLT = \vecu \Gp \GLT
\end{equation}
where $\GLT$ is an $(h \times n)$ matrix. The different constructions addressed in this paper differ in how matrix $\GLT$ is built, as we will explain later in this section.

The output symbols  are transmitted over a \ac{$q$-EC}. At its output each transmitted symbol is either correctly received or erased.\footnote{We remark that, due to the fact that LT output symbols are generated independently of each other, the results developed in this paper remain valid regardless the statistic of the erasures introduced by the channel.}
We denote by $m$ the number of output symbols collected by the receiver, and we express it as $m=k+\absoverhead$, where $\absoverhead$ is the absolute receiver overhead. Let us denote by ${\mathbf{\Rrosymb}=(\Rrosymb_0, \Rrosymb_1, \ldots, \Rrosymb_{m-1})}$ the vector of $m$ received output symbols. Denoting by $\mathcal{I} = \{i_0, i_1, \hdots, i_{m-1} \}$ the set of indices corresponding to the $m$ non-erased symbols, we have $\Rrosymb_j = \Rosymb_{i_j}$. An \ac{ML} decoder proceeds by solving the linear system of equations
\begin{align}
\mathbf{\Rrosymb} = \vecu \Grx
\label{eq:sys_eq}
\end{align}
where
\begin{align}
\Grx = \Gp \GrxLT
\end{align}
and where $\GrxLT$ is the submatrix of $\GLT$ formed by the $m$ columns with indices in $\mathcal{I}$.

\begin{figure}[t]
    \centering
    \includegraphics[width=\columnwidth]{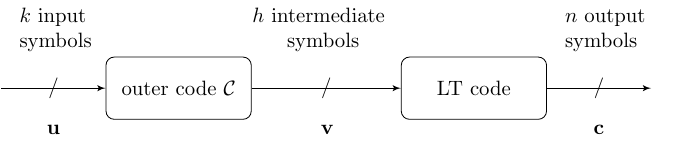}
    \caption{ Block diagram of Raptor encoding.}\label{fig:Raptor_enc}
\end{figure}

\subsection{Raptor Code Constructions} 

The first construction considered in this paper is referred to as \emph{Raptor code over $\GF{q}$}. In this construction each column of $\GLT$ is generated by first randomly drawing an output  degree $d$, according to a probability distribution ${\Omega= (\Omega_1, \Omega_2, \ldots, \Omega_{\dmax})}$, and then by drawing $d$ different indices uniformly at random between $1$ and $h$.
The distribution $\Omega$ is usually referred to as output degree distribution and its generating function is
\[
\Omega(x)=\sum_{d=1}^{d_{\max}} \Omega_d x^d.
\]
Finally, the elements of the column in the row positions corresponding to these indices are drawn independently and uniformly at random from $\mathbb {F}_{q} \backslash \{0\}$, while all other elements of the column are set to zero.

The second considered construction  is referred to as \emph{multi-edge type Raptor code}. This construction is characterized by having two different types of intermediate symbols, namely, type $A$ and type $B$. Thus, the vector of intermediate symbols after reordering can be expressed as $\vecv= (\vecv_A, \vecv_B)$, where $\vecv_A$ and $\vecv_B$ denote the vectors of intermediate symbols of types $A$ and $B$ respectively. Furthermore, we denote the number of intermediate symbols of type $A$ and $B$ as $h_A$ and $h_B$ respectively. We have  $h_A + h_B = h$.
This Raptor code construction is characterized by a relationship between output symbols and intermediate symbols in the form
\begin{equation} \label{eq:met_enc}
\mathbf{\Rosymb} = \vecv \, \GLT = (\vecv_A, \vecv_B) \GLT = (\vecv_A, \vecv_B)
 \left[ \begin{array}{c}
	\GLTA \\ \hline
	\GLTB
\end{array}
\right].
\end{equation}
Under the assumption that $n$ output symbols are generated, $\GLTA$ and  $\GLTB$ have sizes $(h_A \times n)$ and  $(h_B \times n)$ respectively. Each column of $\GLT$ is generated by first drawing two output degrees $j$ and $s$ according to a joint probability distribution $\Omega_{j,s}$ whose bivariate generating function is\footnote{This definition implies  $\Omega_{0,1}=\Omega_{1,0}=0$ (besides $\Omega_{0,0}=0$), which is in line with the distribution used for the RaptorQ code \cite{lubyraptorq}. This assumption is practically motivated but is not strictly necessary.}
\[
\Omega (x,z) = \sum_{j=1}^{h_A} \sum_{s=1}^{h_B} \Omega_{j,s} \, x^j z^s.
\]
For each column, $j$ different indices are drawn uniformly at random in  $\{1, 2, \dots ,h_A\}$ and the elements of the column in $\GLTA$ at the rows corresponding to these indices are drawn independently and uniformly from $\GF{q} \backslash \{0\}$, while all other elements of the column of $\GLTA$ are set to zero. In a similar way, $s$ different indices are picked uniformly at random in $\{1, 2, \dots,  h_B\}$ and the elements of the column in $\GLTB$ at the rows corresponding to these indices are drawn independently and uniformly from $\GF{q} \backslash \{0\}$, while all other elements of the column of $\GLTB$ are set to zero.

The third construction considered is referred to as  \emph{Raptor code over $\GF{q}$ with a $0/1$ LT code}. This construction is relevant to $q>2$, since otherwise it collapses to the first construction. It is similar to the first construction (Raptor code over $\GF{q}$), but all non-zero coefficients of $\GLT$ are equal to $1 \in \GF{q}$. Thus, each column of $\GLT$ is generated by first drawing an output  degree $d$ according to the degree distribution ${\Omega= (\Omega_1, \Omega_2, \ldots, \Omega_{\dmax})}$, and then by picking $d$ different indices uniformly at random in $\{1, 2, \dots,  h\}$. Finally, the elements of the column with rows corresponding to these indices are set to $1$, while all other elements of the column are set to zero.  The relationship between input and output symbols is still given by \eqref{eq:raptor_enc}, where vectors $\mathbf{\Rosymb}$, $\vecv$ and $\vecu$ have elements in $\GF{q}$, matrix $\Gp$ has elements in $\GF{q}$ as well, and the elements of $\GLT$ belong to $\{0,1\} \subset \GF{q}$.
The advantage of this construction is that encoding and decoding complexities are significantly reduced when using a standard computing platform, particularly when $q$ is  a power of $2$.

Finally, the fourth construction considered is referred to as \emph{multi-edge type Raptor code over $\GF{q}$ with a $0/1$ LT code}. As its name indicates this construction is a combination of the second and third constructions described before. In particular, this construction is the same as the second construction, except for the fact that the non-zero elements in $\GLT$, and therefore in $\GLTA$ and $\GLTB$, take always value $1$.

This last construction closely resembles the RaptorQ code \cite{lubyraptorq}, representing the state of art fountain code at the time of writing. 
The RaptorQ code is built over $\GF{256}$. Its outer code is itself obtained as the serial concatenation of two block codes, the first code being a quasi-cyclic nonbinary \ac{LDPC} code and the second code being a nonbinary code defined by a dense parity-check matrix. In particular, the quasi-cyclic \ac{LDPC} code has all its nonzero elements in the parity-check matrix equal to $1 \in \GF{256}$, whereas the second code resembles a random code over $\GF{256}$. The intermediate symbols belong to two different classes, which are called \emph{LT symbols} and \emph{permanently inactive symbols}. The LT code is a $0/1$ LT code characterized by the bivariate degree distribution
\[
\Omega (\x,\z) = \Omega (\x) \left( \frac{\z^2 + \z^3}{2} \right)
\]
where $\x$ and $\z$ are, respectively, the dummy variables associated with LT and permanently inactive symbols, and $\Omega (\x)$ is a degree distribution with maximum output degree $30$. Finally, we remark that the RaptorQ construction can be made systematic.\footnote{A Raptor code is made systematic by adding a further precoding stage and specifying the seed of the pseudorandom generator which is used to generate the LT output symbols, see \cite{shokrollahi2003systematic,Lazaro:phd} for more details.}
Thus, the RaptorQ code in its non-systematic form\footnote{The RaptorQ code is in non-systematic form when random Encoding Symbol Identifiers (ESI) are used \cite{luby2011rfc}. } is an example of the fourth construction considered in this paper (multi-edge type Raptor code over $\GF{q}$ with a $0/1$ LT code). For more details about the RaptorQ construction as well as the design choices involved we refer the reader to \cite{shokrollahi2011raptor}.

\section{Bounds on the Error Probability of Raptor Codes} \label{sec:bounds}

This section contains the main contribution of this paper, a series of bounds on the performance of the different Raptor code constructions presented in Section~\ref{sec:constructions}. Proofs of these bounds are deferred to Section~\ref{sec:proofs}.
The first theorem establishes a bound on the probability of decoding failure of a Raptor code over $\mathbb {F}_{q}$.
\begin{theorem}\label{theorem:rateless}
Consider a Raptor code  over $\mathbb {F}_{q}$ with an $(h,k)$ outer code $\precodegeneric$ characterized by a weight enumerator $\weo$, and an inner \ac{LT} code with output degree distribution $\Omega$.
The probability of decoding failure under \ac{ML} erasure decoding, given that ${k+\absoverhead}$ output symbols have been collected by the receiver, can be upper bounded as
\begin{align}\label{eq:UB_Fq}
\Pf  \leq \frac{1}{q-1} \sum_{l=1}^h \weo_{\l} \pil^{k+\absoverhead}
\end{align}
where  $\pil$ is the probability that the generic output symbol $\rosymb$ is equal to $0$ given that the vector $\vecv$ of intermediate symbols has Hamming weight $l$. The expression of $\pil$ is
\begin{align}
\pil &= \frac{1}{q} +  \frac{q-1}{q} \sum_{j=1}^{\dmax} \Omega_j    \frac{ \krawtchouk{j}{l}{h}{q} }{\krawtchouk{j}{0}{h}{q}}.
\label{eq:pl}
\end{align}
\end{theorem}
The upper bound in Theorem~\ref{theorem:rateless} also applies to LT codes. In that case, $h=k$ and $\weo_{\l}$ is simply the total number of sequences of Hamming weight $l$ and length $k$,
\[
\weo_{\l}= \binom{k}{l} (q-1)^{l}.
\]
The upper bound thus obtained for LT codes coincides with the bound in \cite[Theorem 1]{schotsch:2013}. Theorem~\ref{theorem:rateless}  may be extended to multi-edge type Raptor codes over $\mathbb{F}_q$ as follows.
\begin{theorem}\label{theorem:bound_multi_edge}
Consider a multi-edge type Raptor code over $\mathbb {F}_{q}$ with
an $(h,k)$ outer code $\precodegeneric$ characterized by a bivariate weight enumerator polynomial $\weo (x,z)$ and an inner \ac{LT} code with bivariate output degree distribution $\Omega(x,z)$.
The probability of decoding failure under \ac{ML} erasure decoding given that ${k+\absoverhead}$ output symbols have been collected by the receiver can be upper bounded as
\[
\Pf  \leq \frac{1}{q-1} \sum_{ \substack{0 \leq l \leq h_A \\ 0 \leq t \leq h_B  \\ l+t > 0 }} \weo_{\l,t} \pilt^{k+\absoverhead}
\]
where
\begin{equation}\label{eq:pilt}
\pilt = \frac{1}{q} + \frac{q-1}{q} \sum_{j=1}^{h_A} \sum_{s=1}^{h_B} \Omega_{j,s}  \frac{ \krawtchouk{j}{l}{h_A}{q} }{ \krawtchouk{j}{0}{h_A}{q} }  \frac{ \krawtchouk{s}{t}{h_B}{q} }{ \krawtchouk{s}{0}{h_B}{q} }.
\end{equation}
\end{theorem}

The next result establishes a bound on the probability of decoding failure of a Raptor code over $\mathbb {F}_{q}$ with a $0/1$ LT code.

\begin{theorem}\label{theorem:bound_LTbin}
Consider a Raptor code over $\mathbb {F}_{q}$  with a $0/1$ LT code  having an output degree distribution $\Omega$ and with an $(h,k)$ outer code $\precodegeneric$ characterized by a composition enumerator $\lenum$.
The probability of decoding failure under \ac{ML} erasure decoding given that ${k+\absoverhead}$ output symbols have been collected by the receiver can be upper bounded as
\begin{align}\label{eq:UB_01LT}
\Pf  \leq  \mkern-4mu \frac{1}{q-1} \mkern-4mu \sum_{\labv \neq \comp(\mathbf{0})} \mkern-5mu \lenum  \mkern-2mu \left(  \sum_{j=1}^{\dmax}  \Omega_j  \mkern-7mu \sum_{ \labgv \, \in \Labgvsetj}  B(\labgv) \,
\frac{  \binom{ \lab_0}{ \labg_0} \binom{ \lab_1}{ \labg_1} \mkern-2mu{\cdots} \binom{ \lab_{q-1}}{ \labg_{q-1}} } {\binom{h}{j}} \right)^{ \mkern-8mu k+\delta}
\end{align}
where $\Labgvsetj$ is the set of all possible compositions for vectors in $\mathbb{F}_q^j$.
\end{theorem}
The upper bound in Theorem~\ref{theorem:bound_LTbin} can be extended to the multi-edge type case as follows.
\begin{theorem}\label{lemma:bound_multi_edge_bin_LT}
Consider a multi-edge type Raptor code over $\GF{q}$ with a $0/1$ LT code having bivariate output degree distribution $\Omega(x,z)$, and with
an $(h,k)$ outer code $\precodegeneric$ characterized by a bivariate composition enumerator $\lenumbi$.
The probability of decoding failure under \ac{ML} erasure decoding given that ${k+\absoverhead}$ output symbols have been collected by the receiver can be upper bounded as
\begin{align}\label{eq:met_q_bin}
\Pf  & \leq \frac{1}{q-1} \mkern-20mu \sum_{\substack{\labv_A, \labv_B \\ \labv_A + \labv_B \neq \comp(\mathbf{0})}}  \mkern-20mu \lenumbi  \mkern-4mu \left(  \sum_{j=1}^{h_A} \sum_{s=1}^{h_B}  \Omega_{j,s}   \mkern-10mu \sum_{ \labgv_A  \in \Labgvsetj}  \sum_{ \labgv_B  \in \Labgvsets} \mkern-10mu  B(\labgv_A + \labgv_B) \right.\\
& \mathrel{\phantom{=}} \left. \times  \frac{  \binom{ \lab_{A,0}}{ \labg_{A,0}} \binom{ \lab_{A,1}}{ \labg_{A,1}} \mkern-2mu{\cdots} \binom{ \lab_{A, q-1}}{ \labg_{A, q-1}} } {\binom{h_A}{j}} \,
\frac{  \binom{ \lab_{B,0}}{ \labg_{B,0}} \binom{ \lab_{B,1}}{ \labg_{B,1}} \mkern-2mu{\cdots} \binom{ \lab_{B, q-1}}{ \labg_{B, q-1}} } {\binom{h_B}{s}} \right)^{k+\delta}
\end{align}
where $\Labgvsetj$ and $\Labgvsets$ are the set of all possible compositions for vectors in $\mathbb{F}_{\!q}^j$ and in $\mathbb{F}_{\!q}^s$, respectively.
\end{theorem}

Each of the above theorems specializes the union bound~\eqref{eq:boole} for a specific Raptor construction, providing an explicit expression for the corresponding $S_1$ parameter. By developing an expression for $S_2$, it is also possible to bound the decoding failure probability from below via \eqref{eq:order2_bonferroni} or \eqref{eq:dawson_sankoff}. Hereafter we provide such a lower bound for a Raptor code over $\GF{q}$ with a $0/1$ LT code and, as a particular case, for a Raptor code over $\bbF{2}$. The lower bounds exploit the sets $\cK_{q,h}$ and $\cT_{2,h}$ defined in Section~\ref{sec:results_joint}.
\begin{theorem}\label{th:LB_01LT}
Consider a Raptor code over $\mathbb{F}_q$ with a 0/1 LT code  having output degree distribution $\Omega$, and an $(h,k)$ outer code $\mathcal{C}$ characterized by a composition enumerator $\mathcal{Q}_{\mathbf{f}}$. The probability of decoding failure under ML erasure decoding, given that $ k + \delta$ output symbols have been collected by the receiver, fulfills
\begin{equation}\label{eq:LB_01LT}
\Pf \geq \frac{\theta S_1^2}{(2-\theta)S_1 + 2 S_2} + \frac{(1-\theta)S_1^2}{(1-\theta)S_1+2S_2} \geq S_1 - S_2
\end{equation}
where $\theta = 2 S_2 / S_1 - \lfloor 2 S_2 / S_1 \rfloor$, $S_1$ equals the right-hand side of \eqref{eq:UB_01LT}, and
\begin{align}
S_2 =& \frac{1}{2(q-1)^2} \sum_{\bkappa \in \cK_{q,h}} \JCEF_{\bkappa} \notag \\ 
&\times \Bigg( \sum_{j=1}^{d_{\max}} \Omega_j \sum_{\bupsilon \in \mathsf{\Upsilon}_j} B(\bgamma_1(\bupsilon)) B(\bgamma_2(\bupsilon)) \frac{\prod_{s,t} {\kappa_{s,t} \choose \upsilon_{s,t}}} {{h \choose j}} \Bigg)^{k+\delta}. \label{eq:S2_01LT}
\end{align}
In \eqref{eq:S2_01LT}, $\mathsf{\Upsilon}_j$ is the set of all possible joint compositions for vector pairs in $\bbF{q}^j \times \bbF{q}^j$.\\ Moreover, for $q=2$: (i) the parameter $S_1$ equals the right-hand side of \eqref{eq:UB_Fq} (expressed with $q=2$); (ii) the parameter $S_2$ reduces to
\begin{align}\label{eq:LB_binary}
S_2 = \frac{1}{2} \mkern-3mu \sum_{\btau  \in \cT_{2,h}} \mkern-9mu \JWEF_{\btau} \mkern-6mu \left( \mkern-3mu \sum_{j=1}^{d_{\max}} \mkern-3mu \Omega_j \mkern-7mu \sum_{(i_1,i_2,i_3)} 
\mkern-7mu \frac{{\tau_0 \choose j - i_{1} - i_{2} - i_{3}}{\tau_1 \choose i_{1}}{\tau_2 \choose i_{1}}{\tau_3 \choose i_{3}}}{{h \choose j}}  \right)^{\mkern-6muk+\delta}
\end{align}
where $\JWEF_{\btau}$ is the biweight enumerator of the outer code and where the most inner sum in \eqref{eq:LB_binary} is over all integer triplets ${(i_1,i_2,i_3)}$ such that $i_1 + i_2+i_3=j$; both $i_1+i_3$ and ${i_2+i_3}$ are even; $0 \leq i_1 \leq \min\{\tau_1,j\}$, $0 \leq i_2 \leq \min\{\tau_2,j\}$, ${0 \leq i_3 \leq \min\{\tau_3,j\} }$.
\end{theorem}

Theorems~\ref{theorem:rateless}-\ref{th:LB_01LT} apply to Raptor codes with a given outer code.
Next we extend these results to the case of a random outer code drawn from an ensemble of codes. Specifically, we consider a parity-check based ensemble of outer codes, denoted by $\codensemble$, defined by a random matrix of size $(h - k) \times h$ whose elements belong to $\GF{q}$ (here, $k$ may not coincide with the dimension of a specific code in the ensemble, as it will be discussed later). A linear block code of length $h$ belongs to $\codensemble$ if and only if at least one of the instances of the random matrix is a valid parity-check matrix for it. Moreover, the probability measure of each code in the ensemble is the sum of the probabilities of all instances of the random matrix which are valid parity-check matrices for that code. Note that all codes $\precodegeneric$ in $\codensemble$ are linear, have length $h$, and have dimension $k_\precodegeneric \geq k$.
In the following we use the expression \emph{Raptor code ensemble} to refer to the set of Raptor codes obtained by concatenating an outer code belonging to the ensemble $\codensemble$ with an \ac{LT} code. Given a Raptor code ensemble we define its expected probability of decoding failure as
\begin{align}
\barPf = \Exp_{  \precodegeneric } [ \Pf(\precodegeneric)]
\label{eq:ensemble}
\end{align}
where the expectation is taken over all codes $\precodegeneric$ in the ensemble of outer codes $\codensemble$.

The following corollary extends the result of Theorem~\ref{theorem:rateless} to Raptor code ensembles.
\begin{corollary}\label{corollary:rateless}
Consider a Raptor code ensemble over $\GF{q}$ with an outer code randomly drawn from the ensemble $\codensemble$, characterized by an expected weight enumerator ${\weoensemble= \{ \weoensemble_0,\weoensemble_1,\dots,\weoensemble_h \}}$ and an LT code with degree distribution $\Omega$.  Under \ac{ML} erasure decoding and given that ${k+\absoverhead}$ output symbols have been collected by the receiver, the expected probability of the decoding failure can be upper bounded as
\[
\barPf  \leq   \frac{1}{q-1} \sum_{l=1}^h \weoensemble_{\l}  \pil^{k+\absoverhead} \, .
\]
\end{corollary}

The following three corollaries extend Theorems~\ref{theorem:bound_multi_edge}, \ref{theorem:bound_LTbin}, \ref{corollary:bound_multi_edge_bin_LT} and  to Raptor code ensembles.
\begin{corollary}\label{corollary:bound_multi_edge}
Consider a multi-edge type Raptor code ensemble over $\GF{q}$, whose outer code is randomly drawn from a code ensemble characterized by an expected bivariate weight enumerator polynomial $\weoensemble (x,z)$ and an inner \ac{LT} code with bivariate output degree distribution $\Omega(x,z)$.
The expected probability of decoding failure under \ac{ML} erasure decoding given that ${k+\absoverhead}$ output symbols have been collected by the receiver can be upper bounded as
\[
\barPf  \leq \frac{1}{q-1} \sum_{ \substack{0 \leq l \leq h_A \\ 0 \leq t \leq h_B  \\ l+t > 0 }} \weoensemble_{\l,t} \pilt^{k+\absoverhead}
\]
where $\pilt$ is defined in \eqref{eq:pilt}.
\end{corollary}

\begin{corollary}\label{corollary:bound_rq}
Consider an ensemble of Raptor codes over $\GF{q}$  with a $0/1$ LT code with degree distribution $\Omega$ and where the outer code is randomly drawn from a code ensemble $\codensemble$  characterized by an expected composition enumerator $\lenumensemble$.
The expected probability of decoding failure under \ac{ML} erasure decoding given that ${k+\absoverhead}$ output symbols have been collected by the receiver can be upper bounded as
\begin{align}\label{eq:S1_01LT_ensemble}
\barPf  \leq  \mkern-4mu \frac{1}{q-1} \mkern-4mu \sum_{\labv \neq \comp(\mathbf{0})} \mkern-10mu \, \lenumensemble \mkern-2mu  \left(  \sum_{j=1}^{\dmax}  \Omega_j   \sum_{ \labgv \in \Labgvsetj} \mkern-2mu B(\labgv) \,
\frac{  \binom{ \lab_0}{ \labg_0} \binom{ \lab_1}{ \labg_1} \mkern-2mu{\cdots} \binom{ \lab_{q-1}}{ \labg_{q-1}} } {\binom{h}{j}} \right)^{ \mkern-8mu k+\delta}
\end{align}
where $\Labgvsetj$ is the set of all possible compositions for vectors in $\mathbb{F}_q^j$.
\end{corollary}

\begin{corollary}\label{corollary:bound_multi_edge_bin_LT}
Consider a multi-edge type Raptor code ensemble over $\GF{q}$ with a $0/1$ LT code with bivariate output degree distribution $\Omega(x,z)$ and where the outer code is randomly drawn from an ensemble $\codensemble$ characterized by an expected bivariate composition enumerator $\lenumbiensemble$.
The expected probability of decoding failure under \ac{ML} erasure decoding given that ${k+\absoverhead}$ output symbols have been collected by the receiver can be upper bounded as
\begin{align}\label{eq:met_q_bin_ensemble}
\barPf  & \leq \frac{1}{q-1} \mkern-20mu \sum_{\substack{\labv_A, \labv_B \\ \labv_A + \labv_B \neq \comp(\mathbf{0})}}  \mkern-20mu \lenumbiensemble  \mkern-4mu \left(  \sum_{j=1}^{h_A} \sum_{s=1}^{h_B}  \Omega_{j,s}  \mkern-10mu \sum_{ \labgv_A \in \Labgvsetj} \sum_{ \labgv_B \in \Labgvsets} \mkern-10mu B(\labgv_A + \labgv_B) \right.\\
& \mathrel{\phantom{=}} \left. \times  \frac{  \binom{ \lab_{A,0}}{ \labg_{A,0}} \binom{ \lab_{A,1}}{ \labg_{A,1}} {\cdots} \binom{ \lab_{A, q-1}}{ \labg_{A, q-1}} } {\binom{h_A}{j}} \,
\frac{  \binom{ \lab_{B,0}}{ \labg_{B,0}} \binom{ \lab_{B,1}}{ \labg_{B,1}} {\cdots} \binom{ \lab_{B, q-1}}{ \labg_{B, q-1}} } {\binom{h_B}{s}} \right)^{\mkern-2muk+\delta}
\end{align}
where $\Labgvsetj$ and $\Labgvsets$ are the set of all possible compositions for vectors in $\GF{q}^j$ and in $\GF{q}^s$, respectively.
\end{corollary}
Theorem~\ref{th:LB_01LT} can also be extended to Raptor code ensembles where the outer code is drawn from an ensemble of linear block codes all with the same block length.
\begin{corollary}\label{corollary:LB_ensembles}
Consider an ensemble of Raptor codes over $\mathbb{F}_q$ with a 0/1 LT code with degree distribution $\Omega$, where the outer code is drawn randomly from a code ensemble $\codensemble$ characterized by an expected composition enumerator $\lenum$ and an expected bicomposition enumerator $\avgJCEF_{\bkappa}$. The probability of decoding failure under ML erasure decoding, given that $m$ output symbols have been collected by the receiver, fulfills
\smallskip
\begin{align}
\barPf &\geq \frac{\bar{\theta}\, [ \bar{S}_1(m) ]^2}{(2-\bar{\theta})\bar{S}_1(m) + 2 \bar{S}_2(m)} + \frac{(1-\bar{\theta}) [\bar{S}_1(m)]^2}{(1-\bar{\theta})\bar{S}_1(m) + 2 \bar{S}_2(m)} \\
&\geq 
\bar{S}_1(m) - \bar{S}_2(m) \label{eq:LB_01LT_ensemble}
\end{align}
\medskip
where $\bar{\theta} = 2 \bar{S}_2(m) / \bar{S}_1(m) - \lfloor 2 \bar{S}_2(m) / \bar{S}_1(m) \rfloor$ and
\begin{equation}\label{eq:S1_01LT_ensemble_m}
\bar{S}_1(m) = \mkern-2mu \frac{1}{q-1} \mkern-8mu \sum_{\labv \neq \comp(\mathbf{0})} \mkern-10mu \lenumensemble \mkern-2mu  \left( \mkern-4mu  \sum_{j=1}^{\dmax}  \mkern-7mu \Omega_j \mkern-4mu  \sum_{ \labgv \in \Labgvsetj} \mkern-2mu B(\labgv) 
\frac{  \binom{ \lab_0}{ \labg_0} \binom{ \lab_1}{ \labg_1} \mkern-2mu{\cdots} \binom{ \lab_{q-1}}{ \labg_{q-1}} } {\binom{h}{j}} \mkern-4mu \right)^{\mkern-8mu m}
\end{equation}
\begin{align}\label{eq:S2_01LT_ensemble_m}
\bar{S}_2(m) =& \frac{1}{2(q-1)^2} \sum_{\bkappa \in \cK_{q,h}} \JCEF_{\bkappa} \notag \\ 
&\times \Bigg( \sum_{j=1}^{d_{\max}} \Omega_j \sum_{\bupsilon \in \mathsf{\Upsilon}_j} B(\bgamma_1(\bupsilon)) B(\bgamma_2(\bupsilon)) \frac{\prod_{s,t} {\kappa_{s,t} \choose \upsilon_{s,t}}} {{h \choose j}} \Bigg)^{m}.
\end{align}
In \eqref{eq:S2_01LT_ensemble_m}, $\mathsf{\Upsilon}_j$ is the set of all possible joint compositions for vector pairs in $\bbF{q}^j \times \bbF{q}^j$.\\ Moreover, in the particular case $q=2$ we have ${\bar{S}_1(m) = \sum_{l=1}^h \weoensemble_{\l} \pil^{m}}$ where $\pil$ is given by \eqref{eq:pl} (with $q=2$) and
\begin{equation}\label{eq:LB_binary_ensemble}
\bar{S}_2(m) = \mkern-2mu \frac{1}{2} \mkern-7mu \sum_{\btau  \in \cT_{2,h}} \mkern-15mu \avgJWEF_{\btau} \mkern-7mu \left( \mkern-3mu \sum_{j=1}^{d_{\max}} \Omega_j  \mkern-10mu \sum_{(i_1,i_2,i_3)} 
 \mkern-15mu \frac{{\tau_0 \choose j - i_{1} - i_{2} - i_{3}}{\tau_1 \choose i_{1}}{\tau_2 \choose i_{1}}{\tau_3 \choose i_{3}}}{{h \choose j}}  \mkern-4mu \right)^{\mkern-10mu m} .
\end{equation}
%
In \eqref{eq:LB_binary_ensemble}, $\avgJWEF_{\btau}$ is the average bicomposition enumerator of the outer code ensemble. Furthermore, the most inner sum is over all integer triplets $(i_1,i_2,i_3)$ such that $i_1 + i_2+i_3=j$; both $i_1+i_3$ and $i_2+i_3$ are even; ${0 \leq i_1 \leq \min\{\tau_1,j\}}$, ${0 \leq i_2 \leq \min\{\tau_2,j\}}$, ${0 \leq i_3 \leq \min\{\tau_3,j\}}$.
\end{corollary}
\begin{remark}
Note that the bounds provided in Corollaries \eqref{corollary:rateless} to \eqref{corollary:LB_ensembles} hold also for Raptor code ensembles based on  outer codes of fixed dimension $k$ (e.g., systematic-form generator-based outer code ensembles). The proof for this case is trivial, and follows from the linearity of the expectation. The proofs for the case where the outer code is drawn from a parity-check ensemble require some more care, as illustrated in the following section.
\end{remark}
\section{Derivation of the Bounds}\label{sec:proofs}
This section contains the proofs of the results presented in Section~\ref{sec:bounds}.

\subsubsection{Proof of Theorem~\ref{theorem:rateless}}

The proof follows the same approach as for \cite[Theorem~1]{schotsch:2013}. An \ac{ML} decoder solves the linear system of equations in \eqref{eq:sys_eq}. Decoding fails whenever the system does not admit a unique solution, that is, if and only if  $\rank(\Grx)<k$, i.e., if
${\exists\,  \vecu \in \GF{q}^k \backslash \{ \textbf{0}\} \,\, \text{s.t.} \,\, \vecu \Grx = \textbf{0}}$.
For any two  vectors $\vecu \in \GF{q}^k$ and $\vecv \in \GF{q}^h$, we define  $E_{\vecu}$ as the event $\vecu\Gp \GrxLT = \mathbf{0}$, and $E_{\vecv}$ as the event  $\vecv \GrxLT = \mathbf{0}$. We have
\begin{align}
\Pf & =   \Pr\left\{ \bigcup_{\vecu \in \GF{q}^k \backslash \{ \textbf{0}\}} E_{\vecu}  \right\}  = \Pr\left\{ \bigcup_{\vecv \in \precodegeneric \backslash \{ \textbf{0}\} } E_{\vecv} \right\}
\label{eq:existence}
\end{align}
where we made use of the fact that due to outer code linearity, the all zero intermediate word is only generated by the all zero input vector.

Due to linearity of the outer code, if $\vecv \in \precodegeneric$, then $\beta \vecv \in \precodegeneric$ for any $\beta \in \GF{q} \backslash \{ 0\}$.
Furthermore, for any $\beta \in \GF{q} \backslash \{0\}$, $\vecv\, \GrxLT  = 0$ if and only if $\beta \vecv\, \GrxLT  = 0$.
Thus,
for any two outer codewords $\mathbf{v}_1$ and $\mathbf{v}_2$ such that $\mathbf{v}_1 = \beta \mathbf{v}_2$ for some $\beta \in \GF{q} \setminus \{0\}$, the event $E_{\mathbf{v}_1}$ holds if and only if $E_{\mathbf{v}_2}$ does, and we have ${E_{\mathbf{v}_1} \cup E_{\mathbf{v}_2} = E_{\mathbf{v}_1}}$.
If we take a union bound on \eqref{eq:existence}, this allows us dividing it by a factor $q-1$, leading to
\begin{equation}\label{eq:bound_v}
\Pf \leq \frac{1}{q-1} \sum_{\vecv \in  \mathbb \precodegeneric \backslash \{ \textbf{0}\} }  \Pr \left\{ E_{\vecv} \right\}.
\end{equation}
Defining $\precodegeneric_l$ as 
$\precodegeneric_l = \left\{ \vecv \in \precodegeneric : \wh(\vecv) = l \right\}$, the expression can be developed as
\begin{equation}
\Pf \leq \mkern-3mu  \frac{1}{q-1} \mkern-3mu \sum_{l=1}^{h}  \mkern-3mu \sum_{\vecv \in  \mathbb \precodegeneric_l } \mkern-3mu \Pr \mkern-4mu \left\{ E_{\vecv} \right\} = \frac{1}{q-1} \sum_{l=1}^{h} \weo_{\l} \Pr \left\{ E_{\vecv} | \wh(\vecv)=l  \right\}
\end{equation}
where we made use of the fact that, since the neighbors of an output symbol are chosen uniformly at random, $\Pr \left\{ E_{\vecv} \right\}$ does not depend on the specific vector $\vecv$, but only on its Hamming weight.

Observing that the output symbols are independent of each other, we have
\[
\Pr \left\{ E_{\vecv} | \wh(\vecv)=l \right\} = \pil^{k+\absoverhead}
\]
where $\pil = \Pr \{ y=0 | \wh(\vecv)=l\}$.

Let $J$ and $I$ be discrete random variables representing the number of intermediate symbols which are linearly combined to generate the generic output symbol $y$, and the number of non-zero such intermediate symbols, respectively. Note that $I \leq \min \{ J, \wh(\vecv) \}$.
An expression for $\pil$ may be obtained as
\begingroup
\allowdisplaybreaks
\begin{align*}
\pil &= \sum_{j=1}^{\dmax} \Pr \{ y=0 | \wh(\vecv)=l,J=j \} \Pr \{ J=j | \wh(\vecv)=l \} \\
      &\stackrel{(\mathrm{a})}{=} \sum_{j=1}^{\dmax} \Omega_j \Pr \{ y=0 | \wh(\vecv)=l,J=j \} \\
      &\stackrel{(\mathrm{b})}{=} \sum_{j=1}^{\dmax} \mkern-3mu \Omega_j \mkern-14mu\sum_{i=0}^{\min\{j,l\}} \mkern-14mu \Pr \{ y=0 | I=i \} \! \Pr \{ I=i | \wh(\vecv)=l, \mkern-2mu J=j \}
\end{align*}
\endgroup
where $(\mathrm{a})$ is due to
\[
\Pr \{ J=j | \wh(\vecv)=l \} = \Pr \{ J=j \} = \Omega_j
\]
 and $(\mathrm{b})$  to
\[
 \Pr \{ y=0 | \wh(\vecv)=l, J=j, I=i \} = \Pr \{ y=0 | I=i \}.
\]
  Letting \mbox{$\pifroml = \Pr \{ I=i | \wh(\vecv)=l, J=j \}$}, since the $j$ intermediate symbols are chosen uniformly at random by the LT encoder we have
\begin{align}\label{eq:neighbors}
\pifroml = \frac{ \binom{\l}{i} \binom{h-\l}{j-i} } { \binom{h}{j}} \, .
\end{align}
Let us denote $\Pr \{ y=0 | I=i \}$ by $\qi$ and let us observe that the non-zero elements of $\GrxLT$ are \ac{i.i.d.} and uniformly drawn in $\GF{q} \setminus \{0\}$. On invoking Lemma~\ref{lemma:galois} in  Appendix~ \ref{sec:appendix_sum},\footnote{The proof in Appendix~\ref{sec:appendix_sum} is only valid for fields with characteristic $2$, the case of most interest for practical purposes. The proof of the general case is a simple extension of Lemma~\ref{lemma:galois}.} we have
\begin{align}\label{eq:sum}
\qi =\frac{1}{q} \left( 1 + \frac{(-1)^i}{(q-1)^{i-1}}\right).
\end{align}
We conclude that $\pil$ is given by
\begin{align}
\pil &=  \sum_{j=1}^{\dmax} \Omega_j \!\!  \sum_{i=0}^{\min\{j,l\}} \!\! \pifroml  \, \qi
\end{align}
where $\pifroml$ and $\qi$ are given by \eqref{eq:neighbors} and \eqref{eq:sum}, respectively. Expanding this expression and rewriting it using Krawtchouk polynomials and making use of the Chu-Vandermonde identity, one obtains \eqref{eq:pl}.\footnote{The expression of $\pil$ was derived in \cite{schotsch:2013}, where an upper bound on the performance of LT codes was derived. However, the derivation of $\pil$ in \cite{schotsch:2013} is different from the one we provide in this paper. }
\null\hfill$\blacksquare$

\medskip
We remark that \eqref{eq:bound_v} holds not only for Raptor codes over $\GF{q}$, but also for the other three considered constructions. Hence, \eqref{eq:bound_v} represents the starting point in all subsequent proofs.

\subsubsection{Proof of Theorem~\ref{theorem:bound_multi_edge}}

For this construction we may develop \eqref{eq:bound_v} as
\begin{align}
\Pf  & \leq  \frac{1}{q-1}  \sum_{\substack{0 \leq l \leq h_A \\ 0 \leq t \leq h_B  \\ l+t > 0 } } \sum_{\vecv \in \precodegeneric_{l,t} } \Pr \left\{ E_{\vecv} \right\}
\label{eq:existence_multi_edge}
\end{align}
where $\precodegeneric_{l,t}$ is the set of codewords in $\precodegeneric$ with $l$ non-zero elements in $\vecv_A$ and $t$ non-zero elements in $\vecv_B$, formally $\precodegeneric_{l,t} = \left\{ \vecv = ( \vecv_A, \vecv_B) \in \precodegeneric : \hw(\vecv_A)=l, \hw(\vecv_B)=t\right\}$. Making use of the bivariate weight enumerator of the outer code, we can rewrite \eqref{eq:existence_multi_edge} as
\begin{equation}
\Pf  \leq \frac{1}{q-1} \sum_{ \substack{0 \leq l \leq h_A \\ 0 \leq t \leq h_B  \\ l+t > 0 }} \weo_{l,t} \Pr\{E_{\vecv} | \wh(\vecv_A) =l, \wh(\vecv_B)=t \}
\end{equation}
where we made use of the fact that since the neighbors of an output symbol are chosen uniformly at random, $\Pr \left\{ E_{\vecv} \right\}$ does not depend on the particular vector $\vecv$, but only on its split Hamming weight, $\wh(\vecv_A) =l$ and $\wh(\vecv_B)=t$.

Since output symbols are generated independently of each other
\[
 \Pr\{E_{\vecv} | \wh(\vecv_A) =l, \wh(\vecv_B)=t \} = \pilt ^{k+\absoverhead}
\]
where $\pilt = \Pr \{ y=0 | \wh(\vecv_A) =l, \wh(\vecv_B)=t \}$.

Let $J$ and $I$ be two discrete random variables representing, respectively, the number of intermediate symbols of type $A$ which are linearly combined to generate output symbol $y$, and the number of non-zero such intermediate symbols. Similarly, let $S$ and $D$ be two discrete random variables representing, respectively, the number of intermediate symbols of type $B$ which are linearly combined to generate output symbol $y$, and the number of non-zero such intermediate symbols.
Note that we have $I \leq  \min \{ J, \wh(\vecv_A) \}$ and $D \leq   \min \{ S, \wh(\vecv_B) \}$.
The expression of $\pilt$ can be obtained as

\vspace{-6mm}
\begingroup
\allowdisplaybreaks
\begin{align}
   \pilt &=  \mkern-3mu \sum_{j=1}^{h_A} \mkern-3mu \sum_{s=1}^{h_B} \mkern-3mu \Pr\{y=0 | \wh(\vecv_A)=l, \wh(\vecv_B)=t, J=j, S=s\}\\[-2mm]
  & \mathrel{\phantom{=}}\times \Pr\{ J=j, S=s | \wh(\vecv_A)=l, \wh(\vecv_B)=t\}\\
  & \mkern-8mu \stackrel{(\mathrm{a})}{=}  \sum_{j=1}^{h_A} \sum_{s=1}^{h_B} \Omega_{j,s}  \mbox{\small $\Pr\{y\mkern-1mu=\mkern-1mu 0 | \wh(\vecv_A)\mkern-1mu=\mkern-1mul, \wh(\vecv_B)\mkern-1mu=\mkern-1mut, J\mkern-1mu=\mkern-1muj, S\mkern-1mu=\mkern-1mus \mkern-1mu \} $}
   \\
  & \mkern-8mu \stackrel{(\mathrm{b})}{=} \sum_{j=1}^{h_A} \sum_{s=1}^{h_B} \Omega_{j,s} \sum_{i=0}^{\min(j,l)} \sum_{d=0}^{\min(s,t)}  \Pr\{ y=0 | I=i, D=d\}   \\
  &  \mathrel{\phantom{=}}\times \mbox{ \small $\Pr\{ I=i, D=d | \wh(\vecv_A)=l, \wh(\vecv_B)=t, J=j, S=s\}$ } \\
  & \mkern-8mu \stackrel{(\mathrm{c})}{=}\sum_{j=1}^{h_A} \sum_{s=1}^{h_B} \Omega_{j,s} \sum_{i=0}^{\min(j,l)} \sum_{d=0}^{\min(s,t)}  \Pr\{ y=0 | I=i, D=d\}   \\
  & \mathrel{\phantom{=}} \mkern-2mu \times \mbox{ \small $\Pr\{ I\mkern-1mu=\mkern-1mui | \wh(\vecv_A)\mkern-2mu=\mkern-2mu l, J\mkern-2mu =\mkern-2mu j\} \mkern-2mu
  \Pr\{ D\mkern-2mu =\mkern-2mu d | \wh(\vecv_B)\mkern-2mu = \mkern-2mu t, S\mkern-2 mu=\mkern-2mu s\} $ } 
  \end{align}
 \endgroup
where $(\mathrm{a})$ is due to
\begin{equation}
   \mbox{\small $ \Pr\{ J\mkern-2mu =\mkern-2mu j, S\mkern-2mu = \mkern-2mu s | \wh(\vecv_A)\mkern-1mu =\mkern-1mu l, \wh(\vecv_B)\mkern-1mu =\mkern-1mu t\mkern-1mu \} \mkern-2mu  =\mkern-2mu \Pr\{ J\mkern-1mu =\mkern-1mu j, S\mkern-1mu =\mkern-1mu s \mkern-1mu \} \mkern-2mu=  \mkern-2mu\Omega_{j,s} $}
\end{equation}
$(\mathrm{b})$ is  due to
\begin{align}
  & \mbox{ \small $\Pr \{ y=0 | \wh(\vecv_A)=l, \wh(\vecv_B)=t, J=j, S=s, I=i , D=d\} $} \\
  & \qquad \qquad  \qquad \qquad   = \Pr \{ y=0 | I=i , D=d\}
\end{align}
and $(\mathrm{c})$ follows from independence of $I$ and $D$.
Let us denote ${\Pr \{ y=0 | I=i, D=d \}}$ by $\qiAB$. Since the non-zero elements of $\GrxLT$ are \ac{i.i.d.} and uniformly drawn in $\GF{q} \setminus \{0\}$, on invoking Lemma~\ref{lemma:galois} in the Appendix we have
\begin{align}\label{eq:sumMET}
\qiAB =\frac{1}{q} \left( 1 + \frac{(-1)^{i+d}}{(q-1)^{i+d-1}}\right).
\end{align}
Similarly, letting $\pifromlA = \Pr\{ I=i | \wh(\vecv_A)=l, J=j\}$, we have
\[
\pifromlA = \frac{ \binom{\l}{i} \binom{h_A-\l}{j-i} } { \binom{h_A}{j}}.
\]
If we now define $\pifromlB =  \Pr\{ D=d | \wh(\vecv_B)=t, S=s\}$ and use the same reasoning for the intermediate symbols of type $B$, we have
\[
\pifromlB = \frac{ \binom{t}{d} \binom{h_B-t}{s-d} } { \binom{h_B}{s}}.
\]
Hence, the expression of $\pilt$ is given by
\[
\pilt = \sum_{j=1}^{h_A} \sum_{s=1}^{h_B} \Omega_{j,s} \sum_{i=0}^{\min(j,l)} \sum_{d=0}^{\min(s,t)} \qiAB \, \pifromlA \, \pifromlB
\]
Expanding and rewriting this expression using Krawtchouk polynomials yields \eqref{eq:pilt}.
\null\hfill$\blacksquare$

\subsubsection{Proof of Theorem~\ref{theorem:bound_LTbin}}

Starting again from  \eqref{eq:bound_v} and
defining $\precodegeneric_\labv$ as the set of codewords  with composition $\labv$ in the outer code $\precodegeneric$, i.e., $\precodegeneric_\labv = \left\{ \vecv \in \precodegeneric : \comp(\vecv) = \labv \right\}$,
we have
\begin{align}
\Pf   &\leq \frac{1}{q-1} 
 \sum_{\labv \neq \comp(\mathbf{0}) } \,\,  \sum_{\vecv \in  \mathbb \precodegeneric_\labv } \Pr \left\{ E_{\vecv}  \right\}  \\
& = \frac{1}{q-1} \sum_{\labv \neq \comp(\mathbf{0})}  \lenum  \, \Pr \left\{ E_{\vecv} | \comp(\vecv) = \labv \right\}
\label{eq:existencerq1}
\end{align}
where we made use of the fact that since the neighbors of an output symbol are chosen uniformly at random,
any two codewords having the same composition are characterized by the same probability $\Pr \left\{ E_{\vecv} \right\}$.

Due to independence among the output symbols, we have
\begin{align}
\Pr \left\{ E_{\vecv} | \comp(\vecv) = \labv \right\}  & =  \left(   \Pr \left\{ \Rrosymb=0 | \comp(\vecv) = \labv \right\}  \right)^{k+\delta}.
\label{eq:rq_evf}
\end{align}
Let us now introduce again an auxiliary discrete random variable $J$  to represent the output symbol degree, i.e., the number of intermediate symbols which are summed to generate the generic output symbol $y$. We have

\begin{align}
 \Pr & \left\{ \Rrosymb=0 | \comp(\vecv) = \labv \right\}  =   \sum_{j=1}^{\dmax}  \Omega_j  \Pr \left\{ \Rrosymb=0 | \comp(\vecv) = \labv, J=j \right\}.
\label{eq:rq_y0}
\end{align}

Next, let us introduce the random vector $\Labgv$ representing the composition of the $j$ intermediate output symbols that are added to obtain output symbol $\Rrosymb$.
Recalling that $\Labgvsetj$ is the set of possible compositions of length-$j$ vectors, we can recast $\Pr \left\{ \Rrosymb=0 |  \comp(\vecv) = \labv, J=j \right\}$  as
\begin{align}
\Pr \{ \Rrosymb& =0  | \comp(\vecv) = \labv, J=j \}  \\
&= \sum_{\labgv \in \Labgvsetj}  \Pr \left\{ \Rrosymb=0 | \comp(\vecv) = \labv, J=j, \Labgv=\labgv \,\right\} \\
& \qquad \, \, \, \, \, \times
\Pr \left\{\Labgv=\labgv \,|\comp(\vecv) = \labv , J=j \right\} \\
&= \sum_{\labgv \in \Labgvsetj}  \Pr \left\{ \Rrosymb=0 |\Labgv=\labgv \, \right\}
\Pr \left\{\Labgv=\labgv \, |\comp(\vecv) = \labv , J=j \right\}\\[-4mm]
&= \sum_{ \labgv \in \Labgvsetj}  B(\labgv\,) \, \Pr \left\{\Labgv=\labgv\, |\comp(\vecv) = \labv , J=j \right\}
\label{eq:rq_y0_g}
\end{align}
where the indicator function $B$ has been defined in Section~\ref{sec:prelim}.
The term $\Pr \left\{\Labgv=\labgv \,  | \comp(\vecv) = \labv , J=j \right\}$ can easily be computed making use of
a multivariate hypergeometric distribution. In particular:
\begin{equation}\label{eq:hyperg}
  \Pr\{\Labgv=\labgv \,  | \comp(\vecv) = \labv,  J=j  \} = \frac{  \binom{ \lab_0}{ \labg_0} \binom{ \lab_1}{ \labg_1} \cdots \binom{ \lab_{q-1}}{ \labg_{q-1}} } {\binom{h}{j}}.
\end{equation}
\null\hfill$\blacksquare$

\subsubsection{Proof of Theorem~\ref{lemma:bound_multi_edge_bin_LT}}

The proof tightly follows the proofs of  Theorems~\ref{theorem:bound_multi_edge} and \ref{theorem:bound_LTbin}.
Let us start by defining $\precodegeneric_{\labv_A, \labv_B}$ as the set of codewords in $\precodegeneric$ where $\vecv_A$ and $\vecv_B$ have, respectively, composition $\labv_A$ and $\labv_B$, formally $\precodegeneric_{\labv_A, \labv_B} = \left\{ \vecv = ( \vecv_A, \vecv_B) \in \precodegeneric : \comp(\vecv_A)=\labv_A, \comp(\vecv_B)=\labv_B \right\}$. From \eqref{eq:bound_v} we obtain
\begingroup
\allowdisplaybreaks
\begin{align}
\Pf  &\leq \frac{1}{q-1}
\sum_{\substack{\labv_A, \labv_B \\ \labv_A + \labv_B \neq \comp(\mathbf{0})}} \,\,  \sum_{\vecv \in  \mathbb \precodegeneric_{\labv_A, \labv_B} }  \Pr \left\{  E_{\vecv}  \right\} \\
&=  \frac{1}{q-1} \mkern-20mu \sum_{\substack{\labv_A, \labv_B \\ \labv_A + \labv_B \neq \comp(\mathbf{0})}}
\mkern-20mu \lenumbi  \, \Pr \left\{ E_{\vecv} | \comp(\vecv_A) = \labv_A , \comp(\vecv_B) = \labv_B \right\}.
\label{eq:existencerqmet}
\end{align}
\endgroup
Again we exploited the fact that since the neighbors of an output symbol are chosen uniformly at random, $\Pr \left\{ E_{\vecv} \right\}$
depends only on the split composition of $\vecv$, $\comp(\vecv_A) = \labv_A$ and  $\comp(\vecv_B) = \labv_B$.

Due to independence among the output symbols, we have
\begin{align}
\Pr &\left\{ E_{\vecv} | \comp(\vecv_A) = \labv_A, \comp(\vecv_B) = \labv_B \right\}   \\
& \qquad \qquad  =
\left(   \Pr \left\{ \Rrosymb=0 |\comp(\vecv_A) = \labv_A, \comp(\vecv_B) = \labv_B\right\}  \right)^{k+\delta}.
\label{eq:rq_evf_met}
\end{align}
Introducing the two auxiliary discrete random variables, $J$ and $S$ representing, respectively, the number of intermediate symbols of type $A$ and $B$ which are summed to generate the generic output symbol $y$, we have
\begin{align}
 &\Pr  \left\{ \Rrosymb=0 | \comp(\vecv_A) = \labv_A, \comp(\vecv_B) = \labv_B \right\}  \\
 & =  \mkern-4mu \sum_{j=1}^{h_A}  \mkern-2mu \sum_{s=1}^{h_B} \mkern-8mu  \mbox{ \small $ \Omega_{j,s}  \Pr \left\{ \Rrosymb=0 | \comp(\vecv_A) = \labv_A, \comp(\vecv_B) = \labv_B, J=j , S= s \right\}$}.
\end{align}
Next, let the two random vectors $\Labgv_A$ and $\Labgv_B$ represent, respectively,  the composition of the $j$ intermediate symbols of type $A$ and $s$ intermediate symbols of type $B$ that are added to obtain output symbol $\Rrosymb$. Let us also recall that $\Labgvsetj$ and $\Labgvsets$ represent the set of possible compositions of length-$j$ and $s$ vectors, respectively.
We can recast the rightmost term in the last expression  as
\begingroup
\allowdisplaybreaks
\begin{align}
&  \Pr  \{  \Rrosymb=0 | \comp(\vecv_A) = \labv_A, \comp(\vecv_B) = \labv_B, J=j , S= s  \}  \\
&= \sum_{\labgv_A \in \Labgvsetj} \sum_{\labgv_B \in \Labgvsets} \Pr \{ \Rrosymb=0 |\comp(\vecv_A) = \labv_A, \comp(\vecv_B) = \labv_B, J=j , \\[-2mm]
& \mkern137mu S= s,
 \Labgv_A=\labgv_A,  \Labgv_B=\labgv_B \} \\
& \mathrel{\phantom{=}} \mbox{ \small $ \times \mkern-2mu  \Pr \mkern-2mu \left\{\Labgv_A \mkern-2mu = \mkern-2mu \labgv_A , \Labgv_B \mkern-2mu = \mkern-2mu \mkern-2mu \labgv_B |\comp(\vecv_A) \mkern-2mu = \mkern-2mu \labv_A, \comp(\vecv_B)\mkern-2mu  = \mkern-2mu \labv_B, J\mkern-2mu = \mkern-2mu j , S\mkern-2mu = \mkern-2mu s  \mkern-2mu \right\} $} \\[-2mm]
&= \sum_{\labgv_A \in \Labgvsetj} \sum_{\labgv_B \in \Labgvsets}   \Pr \left\{ \Rrosymb=0 |\Labgv_A=\labgv_A , \Labgv_B=\labgv_B \right\} \\
& \mathrel{\phantom{=}} \mbox{ \small $ \times \mkern-2mu  \Pr \mkern-2mu \left\{\Labgv_A \mkern-2mu = \mkern-2mu \labgv_A , \Labgv_B \mkern-2mu = \mkern-2mu \mkern-2mu \labgv_B |\comp(\vecv_A) \mkern-2mu = \mkern-2mu \labv_A, \comp(\vecv_B)\mkern-2mu  = \mkern-2mu \labv_B, J\mkern-2mu = \mkern-2mu j , S\mkern-2mu = \mkern-2mu s  \mkern-2mu \right\} $} \\[-2mm]
&= \mkern-2mu \sum_{\labgv_A \in \Labgvsetj} \mkern-2mu \sum_{\labgv_B \in \Labgvsets} \mkern-6mu  B(\labgv_A+ \labgv_B) \,
\mkern-1mu \Pr \left\{\Labgv_A = \labgv_A |\comp(\vecv_A) = \labv_A , J=j \right\}
 \mkern-1mu \\[-2mm]
 & \mathrel{\phantom{=}} \times \Pr \left\{\Labgv_B = \labgv_B |\comp(\vecv_B) = \labv_B , S=s \right\}.
\end{align}
\endgroup
The term $\Pr \left\{\Labgv_A=\labgv_A |\comp(\vecv_A) = \labv_A , J=j \right\}$ can easily be computed making use of
a multivariate hypergeometric distribution. Concretely, we have
\begin{align}\label{eq:hyper_metg}
  \Pr\{\Labgv_A=\labgv_A  |\comp(\vecv_A) = \labv_A,  J=j  \} &= \frac{  \binom{ \lab_{A,0}}{ \labg_{A,0}} \binom{ \lab_{A,1}}{ \labg_{A,1}} \cdots \binom{ \lab_{A, q-1}}{ \labg_{A, q-1}} } {\binom{h_A}{j}}
\end{align}
and the same holds for
\begin{align}\label{eq:hyperg_metB}
  \Pr\{\Labgv_B=\labgv_B  |\comp(\vecv_B) = \labv_B,  S=s  \} &= \frac{  \binom{ \lab_{B,0}}{ \labg_{B,0}} \binom{ \lab_{B,1}}{ \labg_{B,1}}\mkern-2mu {\cdots} \binom{ \lab_{B, q-1}}{ \labg_{B, q-1}} } {\binom{h_B}{s}}.
\end{align}
\null\hfill$\blacksquare$

\subsubsection{Proof of Theorem~\ref{th:LB_01LT}}

Applying to the outer codebook the indexing and partition described in Section~\ref{sec:results_joint} we can~write
\begingroup
\allowdisplaybreaks
\begin{align*}
\Pf &= \Pr \left\{ \bigcup_{\vecv \in \precodegeneric \setminus \{\mathbf{0}\} }\, E_{\vecv} \right\} \stackrel{\mathrm{(a)}}{=} \Pr \left\{ \bigcup_{a=1}^{M_{q,k}-1}\, E_{\tilde{\vecv}_a} \right\} \notag \\
&\stackrel{\mathrm{(b)}}{\geq} \sum_{a=1}^{M_{q,k}-1} \Pr \left\{E_{\tilde{\vecv}_a}\right\} - \sum_{0<a<b<M_{q,k}} \Pr \left\{E_{\tilde{\vecv}_a} \cap E_{\tilde{\vecv}_b}\right\} \notag \\
&\stackrel{\mathrm{(c)}}{=} \sum_{a=1}^{M_{q,k}-1} \Pr \left\{E_{\tilde{\vecv}_a}\right\}- \frac{1}{2} \sum_{(s,t)\in\tilde{\cD}_{q,k}} \Pr \left\{E_{\vecv_s} \cap E_{\vecv_t}\right\} \notag \\
&\stackrel{\mathrm{(d)}}{=} \frac{1}{q-1} \mkern-4mu \sum_{\vecv \in \mathcal{C} \setminus \{\mathrm{0}\}} \mkern-10mu \Pr (E_{\vecv}) - \frac{1}{2(q-1)^2} \mkern-12mu \sum_{(s,t)\in\cD_{q,k}} \mkern-14mu \Pr \{E_{\vecv_s} \cap E_{\vecv_t}\}
\end{align*}
\endgroup
where: $\mathrm{(a)}$ is due to the fact that, if two codewords $\vecv$ and $\bfz$ belong to the same part $\mathcal{P}_a$ (i.e., they are linearly dependent), then $E_{\vecv}$ occurs if and only if $E_{\bfz}$ occurs; $\mathrm{(b)}$ is a direct application of degree-two Bonferroni inequality \eqref{eq:order2_bonferroni}; $\mathrm{(c)}$ follows from the definition of $\tilde{\cD}_{q,k}$ given in Section~\ref{sec:results_joint} and from $\Pr\{E_{\vecv_s} \cap E_{\vecv_t} \} = \Pr\{E_{\vecv_t} \cap E_{\vecv_s}\}$ for any $s$ and $t$; $\mathrm{(d)}$ is due the definition of $\cD_{q,k}$ given in Section~\ref{sec:results_joint} and to the fact that, if $\vecv_1$ and $\vecv_2$ belong to some part $\mathcal{P}_a$ and $\bfz_1$ and $\bfz_2$ belong to another part $\mathcal{P}_b$, then $E_{\vecv_1} \cap E_{\bfz_1}$ occurs if and only if $E_{\vecv_2} \cap E_{\bfz_2}$ occurs. The last obtained expression is a degree-two Bonferroni lower bound for $\Pf$ in the form $\Pf \geq S_1 - S_2$. The term $S_1$ has been developed in Theorem~\ref{theorem:bound_LTbin} and equals the right-hand side of \eqref{eq:UB_01LT}. The term $S_2$ can be further developed as
\begin{align}\label{eq:lower_bound_q_LT01}
S_2 &= \frac{1}{2(q-1)^2} \sum_{(s,t)\in\cD_{q,k}} \Pr \{E_{\vecv_s} \cap E_{\vecv_t}\}\\
& \stackrel{\mathrm{(e)}}{=} \mkern-4mu \frac{1}{2(q-1)^2} \mkern-10mu \sum_{\bkappa \in \cK_{q,h}} \mkern-14mu \JCEF_{\bkappa} \mkern-4mu  \Pr \{E_{\vecv} \cap E_{\bfz} | \jc(\vecv,\bfz) = \bkappa\} \notag \\
&\stackrel{\mathrm{(f)}}{=} \mkern-4mu \frac{1}{2(q-1)^2} \mkern-10mu \sum_{\bkappa \in \cK_{q,h}} \mkern-14mu  \JCEF_{\bkappa} \mkern-4mu \left( \Pr \{ \{ y_{\vecv}=0\} \mkern-4mu \cap \mkern-4mu \{y_{\bfz}=0\} |\, \jc(\vecv,\bfz\} \right)^{\mkern-2mu k+\delta} \mkern-3mu.
\end{align}
In the previous equation array, $\mathrm{(e)}$ holds since the probability $\Pr \{E_{\vecv} \cap E_{\bfz}\}$ is the same for all codeword pairs $(\vecv, \bfz)$ with the same bicomposition. In $\mathrm{(f)}$ we have denoted by $y_{\vecv}$ the output symbol given that $\vecv$ is the intermediate codeword and we have exploited independence of output symbols.

Next, let the random variable $J$ represent the output symbol degree. Moreover, for given bicomposition $\jc(\vecv,\bfz)=\bkappa$ and given $J=j$, define $\bUpsilon$ as the joint composition of the the two vectors in $\bbF{q}^j$ representing the $j$ symbols selected in $\vecv$ and $\bfz$. %
We have
\begingroup
\allowdisplaybreaks
\begin{align}\label{eq:LB_pi_q_01LT}
\Pr & \{ \{y_{\vecv}=0\} \cap \{y_{\bfz}=0 \}|\, \jc(\vecv,\bfz) = \bkappa\} \\
& = \sum_{j=1}^{d_{\max}} \Omega_j \Pr \{ \{y_{\vecv}=0\}\cap\{y_{\bfz}=0\} |\, \jc(\vecv,\bfz) = \bkappa, J=j\} \notag \\
& = \sum_{j=1}^{d_{\max}} \Omega_j \sum_{\bupsilon \in \mathsf{\Upsilon}_j} \Pr \{ \{ y_{\vecv}=0\} \cap \{y_{\bfz}=0\} | \bUpsilon = \bupsilon \} \\
& \qquad \qquad \qquad \, \, \, \times \Pr \{ \bUpsilon=\bupsilon | J=j, \jc(\vecv,\bfz) = \bkappa \} \notag \\
&= \sum_{j=1}^{d_{\max}} \Omega_j \sum_{\bupsilon \in \mathsf{\Upsilon}_j} B(\bgamma_1(\bupsilon)) B(\bgamma_2(\bupsilon)) \frac{\prod_{\substack{0\leq s\leq q-1\\0\leq t\leq q-1}} {\kappa_{s,t} \choose \upsilon_{s,t}}} {{h \choose j}} \, .
\end{align}
\endgroup
where $\bgamma_1(\bupsilon)$ and $\bgamma_2(\bupsilon)$, defined in \eqref{eq:gamma_one} and \eqref{eq:gamma_two}, are the compositions corresponding to $\bupsilon$. Expression \eqref{eq:S2_01LT} is obtained by substituting \eqref{eq:LB_pi_q_01LT} into \eqref{eq:lower_bound_q_LT01}. The two bounds in \eqref{eq:LB_01LT} then follow as a direct application of degree-two Bonferroni and Dawson-Sankoff bounds, and from the observation that Dawson-Sankoff bound is tighter than the $S_1-S_2$ one.

For $q=2$, the right-hand sides of \eqref{eq:UB_Fq} and \eqref{eq:UB_01LT} coincide. The $S_1$ term is therefore equal to right-hand side of \eqref{eq:UB_Fq} expressed with $q=2$. Next, recall from Remark~\ref{remark:binary_equiv} that for $q=2$ there is a one-to-one correspondence between joint compositions and joint weights. With this correspondence in mind we can write $\JCEF_{\bkappa} = \JWEF_{\btau}$. Again owing to this correspondence, we can establish a bijection between the set of joint compositions $\cK_{2,k}$ and the set of joint weights $\cT_{2,h}$. The right-hand side of \eqref{eq:S2_01LT} may thus be recast as
\begin{align*}
\frac{1}{2} \mkern-8mu \sum_{\btau  \in \cT_h} \mkern-16mu  \mbox{ \small $\JWEF_{\btau} $} \mkern-4mu \Bigg( \mkern-6mu \sum_{j=1}^{d_{\max}} \mkern-6mu \Omega_j \mkern-10mu \sum_{\bupsilon \in \mathsf{\Upsilon}_j} \mkern-16mu
\mbox{ \small $ B(\bgamma_1(\bupsilon)) B(\bgamma_2(\bupsilon)) \frac{{\tau_0 \choose \upsilon_{0,0}} {\tau_1 \choose \upsilon_{0,1}} {\tau_2 \choose \upsilon_{1,0}} {\tau_3 \choose \upsilon_{1,1}}} {{h \choose j}} \mkern-4mu \Bigg)^{\mkern-8mu k+\delta}  $}
\end{align*}
which yields the statement by simply letting
\begin{align*}
\bupsilon = \left[ \begin{array}{cc} j-i_1-i_2-i_3 & \,\,i_1 \\ i_2 & \,\,i_3 \end{array} \right] .
\end{align*}
\null\hfill$\blacksquare$

\subsubsection{Proof of Corollary~\ref{corollary:rateless}}
Due to Theorem~\ref{theorem:rateless}  we may write
\begin{align}
\barPf \leq \Exp_{ \precodegeneric} \Bigg[ \frac{1}{q-1} \sum_{l=1}^h \weo_{\l}(\precodegeneric)  \pil^{k_\precodegeneric+\absoverhead} \Bigg].
\label{eq:ensemble2}
\end{align}
For all outer codes $\precodegeneric \in \codensemble$ we have $k_\precodegeneric \geq k$. Since $\pil \leq 1$  we can write $\pil^{k_\precodegeneric+\absoverhead} \leq \pil^{k+\absoverhead}
$ which allows us to upper bound \eqref{eq:ensemble2} as
\[
\barPf \leq \Exp_{  \precodegeneric} \Bigg[ \frac{1}{q-1} \sum_{l=1}^h \weo_{\l}(\precodegeneric) \pil^{k+\absoverhead} \Bigg]= \frac{1}{q-1} \sum_{l=1}^h \weoensemble_{\l}  \pil^{k+\absoverhead}
\]
where the last equality follows from linearity of expectation.

\null\hfill$\blacksquare$

The proofs of Corollaries \ref{corollary:bound_multi_edge}, \ref{corollary:bound_rq} and \ref{corollary:bound_multi_edge_bin_LT} follow closely that of Corollary 1. Thus, they are omitted for the sake of brevity.

\subsubsection{Proof of Corollary~\ref{corollary:LB_ensembles}}
Let $m$ be the number of symbols collected by the receiver. Denote by $\precodegeneric$ the generic outer code in the ensemble. Denote by $S_1(\precodegeneric,m)$ and $S_2(\precodegeneric,m)$ the parameters $S_1$ and $S_2$ for code $\precodegeneric$ for a fixed number $m$ of collected symbols. Using \eqref{eq:pre_dawson_sankoff} we can write
\begin{align*}
\barPf &= \sum_{\precodegeneric \in \ensemble} \Pr\{\precodegeneric\} \Pf(\precodegeneric) \\
& \geq \sum_{\precodegeneric \in \ensemble} \Pr\{\precodegeneric\} \Big[\frac{2}{r+1} S_1(\precodegeneric,m) - \frac{2}{r(r+1)} S_2(\precodegeneric,m)\Big] \notag \\
&= \frac{2}{r+1} \Exp_{\precodegeneric}[S_1(\precodegeneric,m)] - \frac{2}{r(r+1)} \Exp_{\precodegeneric}[S_2(\precodegeneric,m)] \\
& = \frac{2}{r+1} \bar{S}_1(m) - \frac{2}{r(r+1)} \bar{S}_2(m)
\end{align*}
for any $r \in \{1,\dots,M_{q,k}\}$, where $\bar{S}_1(m)$ and $\bar{S}_2(m)$ are given by \eqref{eq:S1_01LT_ensemble_m} and \eqref{eq:S2_01LT_ensemble_m}, respectively. Taking $r=1$ we obtain the looser bound in \eqref{eq:LB_01LT_ensemble} (i.e., $\barPf \geq \bar{S}_1(m)-\bar{S}_2(m)$). Maximization with respect to $r$ leads us to the tighter bound in \eqref{eq:LB_01LT_ensemble}. (The calculation is the same as that used in \cite{dawson67:inequality} to obtain \eqref{eq:dawson_sankoff} from \eqref{eq:pre_dawson_sankoff} via maximization with respect to $r$.)\footnote{In the extension of the upper bounds to Raptor ensembles, we expressed the number of collected symbols at the receiver as $k_{\precodegeneric}+\delta$ for each randomly drawn outer code $\precodegeneric$, i.e., we considered a fixed absolute overhead with respect to the outer code dimension. In the extension of the lower bounds, instead, the number of collected symbols was expressed as a fixed $m$ for all outer codes. Note that we can also write $\barPf = \bar{\mathsf{P}}_{\mathsf{F}|k_{\precodegeneric}=k} \Pr\{k_{\precodegeneric}=k\} + \bar{\mathsf{P}}_{\mathsf{F}|k_{\precodegeneric}>k} \Pr\{k_{\precodegeneric}>k\}$. Since $\Pr\{k_{\precodegeneric}=k\} <1$ and $\bar{\mathsf{P}}_{\mathsf{F}|k_{\precodegeneric}>k}<1$, we obtain $\bar{\mathsf{P}}_{\mathsf{F}|k_{\precodegeneric}=k} > \barPf - \Pr\{k_{\precodegeneric}>k\}$. If $\Pr\{k_{\precodegeneric}>k\}$ 
is small compared to $\barPf$ (as an example, for a linear random outer code defined by $m$ equations we have $\Pr \{k_{\precodegeneric} > k\} < 2^{-(h-m)} $) then Corollary~\ref{corollary:LB_ensembles} with $m=k+\delta$ may be regarded as an approximate lower bound for the average error probability when the outer code ensemble is expurgated from all codes with dimension larger than $k$.}\null\hfill$\blacksquare$
\section{Error Exponent Analysis}\label{sec:errexp}

In this section, we aim at deriving an error exponent analysis of Raptor code. In particular, a lower bound to the error exponent is obtained for Raptor code ensembles as a function of the outer code ensemble weight spectral shape and of the inner \ac{LT} code distribution. The focus in on both binary and nonbinary Raptor codes.\footnote{The analysis of Raptor code ensemble sequences over $\GF{q}$ with $0/1$ \ac{LT} codes is omitted due to the lack of a definition of an equivalent of the weight spectral shape for (bivariate) composition enumerators.} Before proceeding with the derivation, we need to introduce a few definitions.

Following the definitions of Section \ref{subsec:growthrate} above, we refer to a Raptor code ensemble sequence as a sequence of Raptor code ensembles indexed by the code dimension $k$, where the $k$th Raptor code ensemble is defined by an outer code ensemble $\codensemble_k$ and an inner \ac{LT} code with degree distribution $\Omega(x)$, both over $\GF{q}$.
To emphasize the role of the code dimension, we re-write next \eqref{eq:ensemble} as
$\barPf^{(k)} = \Exp_{  \precodegeneric } [ \Pf(\precodegeneric)]$
where the average is over the outer code ensemble $\codensemble_k$. For a given relative overhead $\epsilon=\delta/k$, with $\epsilon \geq 0$, the error exponent of the Raptor code ensemble sequence is 
\begin{align}
    E(\epsilon)=\lim_{k\rightarrow \infty} -\frac{1}{k}\log_2 \barPf^{(k)}(\epsilon).
    \label{eq:errorexpdef}
\end{align}
Before proceeding with the derivation of a lower bound to the error exponent for general Raptor code ensemble sequences, we illustrate the case of linear random fountain codes as an example.
\begin{example}
The probability of decoding failure for a dimension-$k$ linear random fountain code over $\GF{q}$ can be tightly upper bounded as \cite{Liva10:fountain}
\[
\barPf^{(k)}<\frac{1}{q-1}q^{-\epsilon k}.
\]
For linear random fountain codes we hence have
\begin{align}
    E(\epsilon)&=\lim_{k\rightarrow \infty} -\frac{1}{k}\log_2 \barPf^{(k)}(\epsilon)\\
    &>\lim_{k\rightarrow \infty} -\frac{1}{k}\log_2 \left(\frac{1}{q-1}q^{-\epsilon k}\right)\\
    &=\epsilon \log_2 q.
    \label{eq:errorexp_LRFC}
\end{align}
Note that \eqref{eq:errorexp_LRFC} is positive for positive $\epsilon$, i.e., a positive relative overhead is sufficient to achieve an exponential (in $k$) decay of the decoding failure probability. 
\end{example}

For general Raptor code ensemble sequences, the following theorem provides a lower bound to the error exponent (under mild conditions on the outer code ensemble sequence).

\begin{theorem}
Consider a Raptor code ensemble sequence over $\GF{q}$ defined by an outer code ensemble sequence $\left\{\codensemble_k\right\}$ and an inner \ac{LT} code degree distribution $\Omega(x)$. Let the outer code ensemble sequence spectral shape $G(\omega)$ be well-defined in $[0,1]$. 
If $\frac{1}{h} \log_2 \weoensemble^{(k)}_{\lfloor\omega h\rfloor}\xrightarrow{\mathsf{u}}G(\omega)$ then the Raptor code ensemble sequence error exponent can be lower bounded as
\begin{align}
    E(\epsilon)&\geq - \sup_{\omega \in (0,1]} \left[\frac{1}{R} G(\omega) + (1+\epsilon)\log_2 \varrho_\omega \right] \label{eq:errexplb}
\end{align}   
where $\varrho_\omega=\frac{1}{2}\sum_{j=1}^{\dmax}\Omega_j\left[1-(1-2\omega)^j\right]$.
\end{theorem}
\begin{proof}
For a general Raptor code ensemble sequence, we re-write  the upper bound of on the decoding failure probability from Corollary \ref{corollary:rateless} as
\[
\barPf^{(k)}  \leq   \frac{1}{q-1} \sum_{\omega \in \mathcal{F}_h} \weoensemble^{(k)}_{\lfloor\omega h\rfloor}  \pi_{\lfloor\omega h\rfloor}^{k(1+\epsilon)}
\]
where $\mathcal{F}_h=\left\{\frac{l}{h}\right\}$ with $l=1,\ldots,h$. Following \eqref{eq:errorexpdef}, we have that
\begingroup
\allowdisplaybreaks
\begin{align}
    E(\epsilon) \mkern-4mu &=\lim_{k\rightarrow \infty} -\frac{1}{k}\log_2 \barPf^{(k)}(\epsilon) \\
    & \geq \lim_{h\rightarrow \infty} \mkern-6mu -\frac{1}{hR}\log_2 \frac{1}{q-1} \mkern-8mu \sum_{\omega \in \mathcal{F}_h} \weoensemble^{(hR)}_{\lfloor\omega h\rfloor}  \pi_{\lfloor\omega h\rfloor}^{hR(1+\epsilon)} \\
    &=\lim_{h\rightarrow \infty} \mkern-6mu -\frac{1}{hR}\log_2 \frac{1}{q-1}  \mkern-8mu \sum_{\omega \in \mathcal{F}_h} \mkern-5mu 2^{\log_2 \weoensemble^{(hR)}_{\lfloor\omega h\rfloor} + hR(1+\epsilon)\log_2 \pi_{\lfloor\omega h\rfloor}} \\
    &\geq \lim_{h\rightarrow \infty} \mkern-6mu -\frac{1}{hR}\log_2 \mkern-4mu \left\{ \mkern-4mu h \sup_{\omega \in \mathcal{F}_h} \mkern-6mu \left[ \mkern-3mu 2^{\log_2 \weoensemble^{(hR)}_{\lfloor\omega h\rfloor} + hR(1+\epsilon)\log_2 \pi_{\lfloor\omega h\rfloor}} \mkern-4mu\right] \mkern-4mu \right\}\\[-1mm]
    &= - \lim_{h\rightarrow \infty}\sup_{\omega \in \mathcal{F}_h}\left[ \frac{1}{R}\log_2 \weoensemble^{(hR)}_{\lfloor\omega h\rfloor} + (1+\epsilon)\log_2 \pi_{\lfloor\omega h\rfloor}\right]\\ 
    &= - \lim_{h\rightarrow \infty}\sup_{\omega \in (0,1]}\left[ \frac{1}{R}\log_2 \weoensemble^{(hR)}_{\lfloor\omega h\rfloor} + (1+\epsilon)\log_2 \pi_{\lfloor\omega h\rfloor}\right].
    \label{eq:errorexpgen1}
\end{align}
\endgroup
If $\frac{1}{h} \log_2 \weoensemble^{(hR)}_{\lfloor\omega h\rfloor}$ converges uniformly to $G(\omega)$ in $[0,1]$, by observing that $ \pi_{\lfloor\omega h\rfloor} \xrightarrow{\mathsf u} \varrho_\omega$ (see \cite[Sec. III]{lazaro:JSAC}), the order of the  limit and the supremum operations in \eqref{eq:errorexpgen1} can be inverted, yielding \eqref{eq:errexplb}.
\end{proof}
\begin{remark}
Observe that the error exponent lower bound is monotonically increasing with $\epsilon$. Let us assume next that, for a given Raptor code ensemble sequence, there exist an $\epsilon^\star>0$ s.t. the right-hand side of \eqref{eq:errexplb} is strictly positive for all $\epsilon>\epsilon^\star$. We can conclude that  the Raptor code ensemble sequence is characterized by a decoding failure probability that decays exponentially fast in $k$ for $\epsilon>\epsilon^\star$. The value of $\epsilon^\star$ can be regarded as an upper bound on the \ac{ML} decoding threshold of the Raptor code ensemble. It is important to stress that this bound on the \ac{ML} decoding threshold may not be tight since it does not capture the performance in the region $\epsilon\leq\epsilon^\star$. In this latter region, the decoding failure probability may still become vanishing small as $k$ grows large at a sub-exponential rate (e.g., only polynomially-fast in $k$).
\end{remark}

\section{Examples of Application to Raptor Codes and Raptor Code Ensembles}\label{sec:numres}

In this section, we apply the results of Sections \ref{sec:bounds} and \ref{sec:errexp} to Raptor codes and Raptor code ensembles.
For the analysis, we use the \ac{LT} output degree distribution employed by standard R10 Raptor codes \cite{MBMS16:raptor,luby2007rfc}, given by
\begin{align}
\Omega_{\mathsf{A}}(\x) &= 0.0098\x + 0.4590\x^2+ 0.2110\x^3+0.1134\x^4   \\
& \mathrel{\phantom{=}}+ 0.1113\x^{10} + 0.0799\x^{11} + 0.0156\x^{40}.
\label{eq:dist_mbms}
\end{align}
\subsection{Raptor Code over $\GF2$ with a Hamming Outer Code}
Consider a binary Raptor code over $\GF2$ with a Hamming outer code.
The weight enumerator of a binary Hamming code of length $h=2^t-1$ and dimension $k=h-t$ can be derived easily using the  recursion $(i+1)\, A_{i+1} + A_i + (h-i+1)\, A_{i-1}= \binom{h}{i}$ with $A_0=1$ and $A_1=0$ \cite{MacWillimas77:Book}. The weight distribution obtained from this recursion can then be incorporated in Theorem~\ref{theorem:rateless} to derive the corresponding upper bound on the failure probability. 
The lower bounds established by Theorem~\ref{th:LB_01LT} (binary case) can also be derived, by employing the Hamming code biweight enumerator, an expression of which was developed in \cite{MacWillimas77:Book}.

\begin{figure}[t]
    \centering
    \includegraphics[width=\columnwidth]{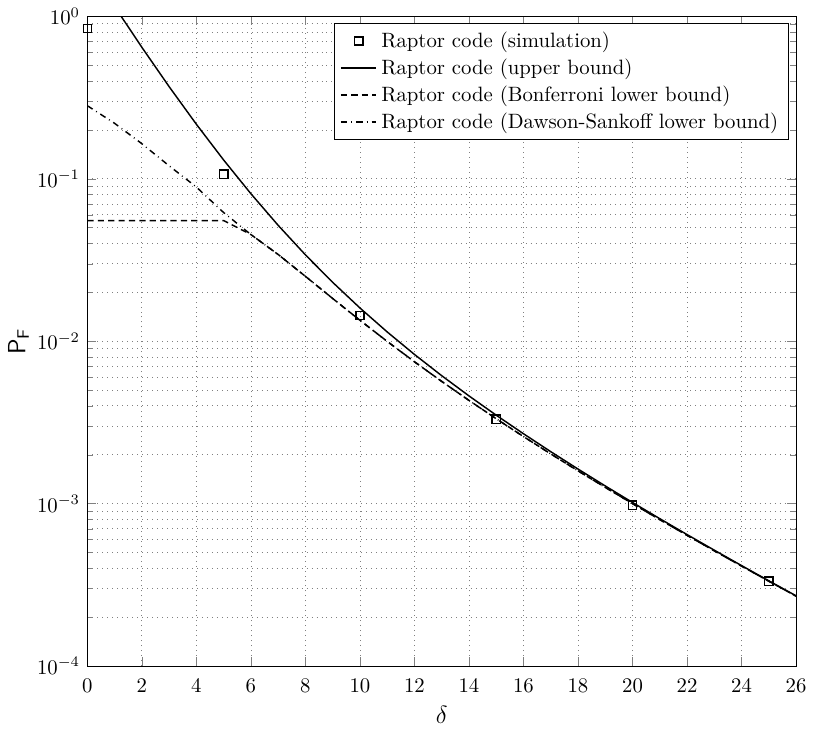}
    \caption[Probability of Decoding failure $\Pf$ vs the absolute overhead for a Raptor with a $(63,57)$ Hamming outer code]{Decoding failure probability $\Pf$ versus the absolute overhead $\delta$ for a binary Raptor code with a $(63,57)$ Hamming outer code and LT distribution $\Omega_{\mathsf{A}}(x)$. Markers: simulation results. Solid: upper bound (Theorem~\ref{theorem:rateless}). Dotted: Degree-two Bonferroni lower bound (Theorem~\ref{th:LB_01LT}). Dot-dashed: Dawson-Sankoff lower bound (Theorem~\ref{th:LB_01LT}).}\label{fig:Hamming_sim}
\end{figure}

Fig.~\ref{fig:Hamming_sim} shows the decoding failure rate for a Raptor code over $\GF2$ employing a $(63,57)$ binary Hamming outer code as a function of the absolute overhead, $\absoverhead$, together with the upper bound from Theorem~\ref{theorem:rateless} and the lower bounds from Theorem~\ref{th:LB_01LT} (binary case).
In order to obtain the values of failure rate, Monte Carlo simulations were run for each $\absoverhead$ until $200$ errors were collected using inactivation decoding. It can be observed how the upper bound is very tight and how the gap between the upper and lower bounds is very small already for values of $\delta$ in the order of $10$. Interestingly, the order-two Bonferroni and the Dawson-Sankoff bounds are practically coincident for $\delta \geq 6$ while for $\delta<6$ the Dawson-Sankoff bound turns to be remarkably tighter.\footnote{The difference $S_1-S_2$ is actually increasing for $\delta \in \{0,\dots,5\}$, it reaches a maximum at $\delta=5$ and then decreases. For $\delta \in \{0,\dots,4\}$ the difference is even negative. However, since the failure probability cannot increase as $\delta$ increases, we can apply the value taken by $S_1-S_2$ at $\delta=5$ to all $\delta < 5$. In contrast,  Dawson-Sankoff bound decreases monotonically over the whole range of $\delta$.}

\subsection{Raptor Code Ensembles with Linear Random Outer Codes}
Next, consider a Raptor code ensemble over $\GF{q}$, with \ac{LT} degree distribution $\Omega_{\mathsf{A}}(x)$ and in which the outer code is picked from the uniform parity-check ensemble, with parity-check matrix of size $(h-k) \times h$ and characterized by \ac{i.i.d.} entries with uniform distribution in $\GF{q}$. The expected weight enumerator for an outer code drawn randomly in $\codensemble$ 
is known to be 
$\weoensemble_{\l} = \binom{h}{\l} q^{-(h-k)}  (q-1)^l$. 
The expected composition enumerator can be obtained from the expected weight enumerator, as discussed in Appendix~\ref{sec:appendix_composition}, while the expected bicomposition enumerator can be obtained as shown in Appendix~\ref{app:bicomp}. 

To obtain the experimental values of the expected decoding failure rate, 
$6000$ different outer codes were generated. For each outer code and for each overhead value, $1000$ inactivation decoding attempts were carried out. The average  failure rate was calculated by averaging the failure rates of the individual Raptor codes. 
To generate an outer code, an $(h-k)\times h$ parity-check matrix was drawn randomly by picking its elements independently and uniformly in~$\GF{q}$.

In Fig.~\ref{fig:pf_k_64_h70} we show simulation results for $k=64$ and $h=70$. Three different Raptor code ensembles were considered, one constructed over $\GF{2}$,  one constructed over $\GF{4}$, and one constructed over $\GF{4}$ with a $0/1$ LT code. We can observe how in all cases the upper bounds are tight, even for small values of $\absoverhead$.
Comparing the two ensembles over $\GF{4}$, it is remarkable that employing a $0/1$ LT code results only in a small performance degradation, which vanishes as $\absoverhead$ increases. Both order-two Bonferroni and Dawson-Sankoff lower bounds are displayed for the binary ensemble. Again, the Dawson-Sankoff bound turns out to be remarkably tighter for small~$\delta$.
\begin{figure}[t]
    \centering
    \includegraphics[height=7.75cm]{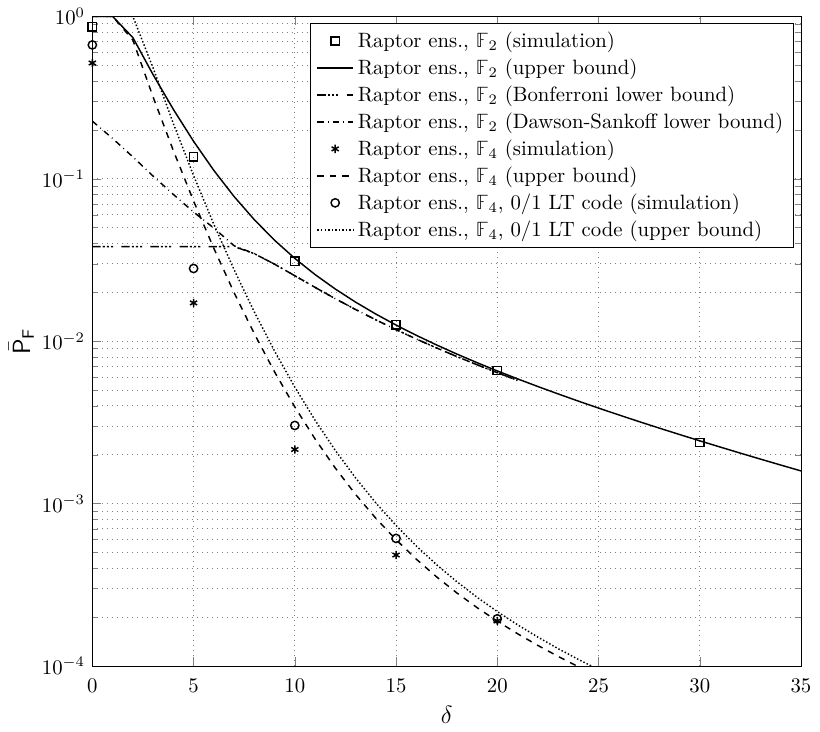}
    \caption{Expected probability of decoding failure $\barPf$ vs absolute overhead for Raptor code ensembles where the outer code is drawn randomly from the uniform parity-check ensemble with $k=64$ and $h=70$. LT distribution: $\Omega_{\mathsf{A}}(x)$.  Lines: upper and lower bounds. Markers: simulation results.}\label{fig:pf_k_64_h70}
\end{figure}

In Fig.~\ref{fig:eexp} lower bounds on the error exponents of various binary Raptor code ensemble sequences are provided. The Raptor code ensemble sequences are defined by the degree distribution $\Omega_{\mathsf{A}}(\x)$ and linear random outer code sequences with (outer) code rates $R=0.90,0.95$ and $0.98$. When the outer code is picked from a binary linear random code ensemble, the error exponent lower bound of \eqref{eq:errexplb} reduces to 
\[
    E(\epsilon) \geq - \sup_{\omega \in (0,1]} \left[\frac{H_b(\omega) +R-1}{R}+ (1+\epsilon)\log_2 \varrho_\omega \right]
\]
where $H_b(\omega)=-\omega \log_2 \omega - (1-\omega) \log_2 (1-\omega)$ is the binary entropy function. The error exponent lower bound for linear random fountain codes of \eqref{eq:errorexp_LRFC} is provided as a reference. As intuition suggests, the error exponent  lower bound for Raptor codes approaches the one of linear random fountain codes as the outer code rate decreases. The upper bounds on the \ac{ML} decoding thresholds are $\epsilon^\star\approx 6 \times 10^{-2}$ for $R=0.98$, $\epsilon^\star\approx1.33\times 10^{-2}$ for $R=0.95$, and $\epsilon^\star\approx 5\times 10^{-4}$ for $R=0.90$.

\begin{figure}[t]
    \centering
    \includegraphics[height=7.75cm]{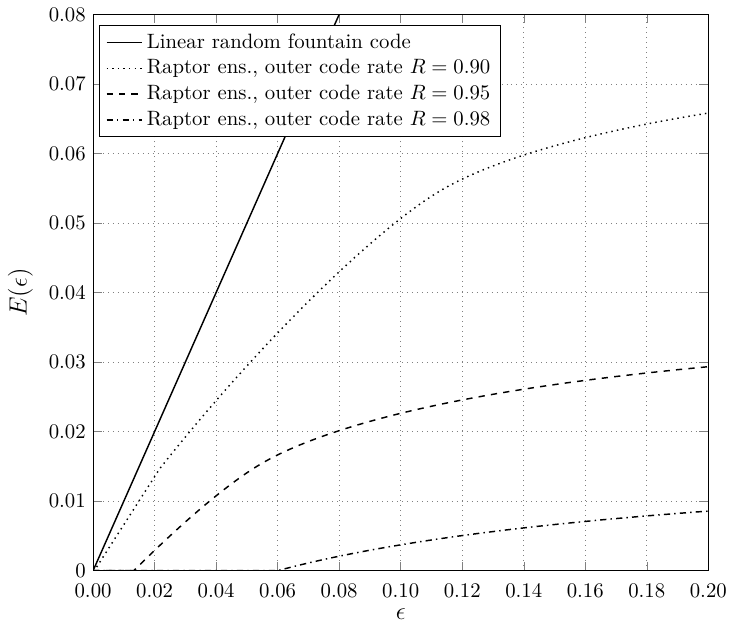}
    \caption{Lower bounds on the error exponent vs. relative overhead $\epsilon$ for binary Raptor code ensemble sequences defined by the degree distribution $\Omega_{\mathsf{A}}(\x)$ and linear random outer code sequences with (outer) code rates $R=0.90,0.95$ and $0.98$. The error exponent lower bound for linear random fountain codes of \eqref{eq:errorexp_LRFC} is provided as reference.}\label{fig:eexp}
\end{figure}

\subsection{Raptor Code Ensembles with Regular LDPC Outer Codes}
We now consider ensembles of Raptor codes in which the outer code is drawn from a $(d_v, d_c)$ regular \acf{LDPC} code ensemble, 
where $d_v$ and $d_c$ are the variable and check node degrees, respectively. In order to draw a code from this ensemble we first generate a random permutation of the $h d_v = (h-k) d_c$ edges between check and  variable nodes. Then we assign to each edge a non-binary label picked uniformly at random in $\GF{q} \backslash \{ 0\}$. The average weight enumerator for this ensemble is reviewed in Appendix~\ref{sec:appendix_composition}, where an expression of its expected composition enumerator is also derived. 

In order to simulate the average probability of decoding failure of the ensemble, $10000$ different outer codes were generated. For each outer code and overhead value, $100$ decoding attempts were carried out. The average probability of decoding failure was obtained averaging the probabilities of decoding failure obtained with the different outer codes.

\begin{figure}[t]
    \centering
    \includegraphics[height=7.75cm]{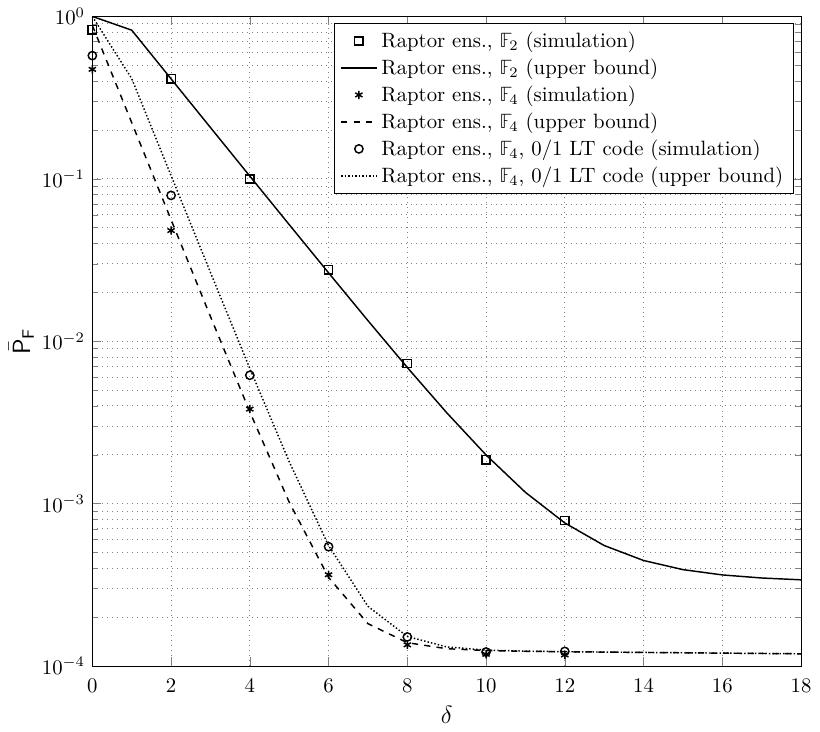}
    \caption{Average probability of decoding failure $\Pf$ vs absolute overhead for two Raptor code ensembles where the outer code is randomly drawn from the $(d_v=3,d_c=15)$ regular LDPC ensemble  with $k=1000$ input symbols and $h=1250$ intermediate symbols. LT distribution: $\Omega_{\mathsf{A}}(x)$. Lines: upper bounds. Markers: simulation results.}\label{fig:LDPC}
\end{figure}

Fig.~\ref{fig:LDPC} shows the average probability of decoding failure for three ensembles of Raptor codes where the outer code is randomly drawn from the $(d_v=3,d_c=15)$ regular LDPC ensemble with $k=1000$ input symbols and $h=1250$ intermediate symbols. The first ensemble is constructed over  $\mathbb{F}_2$, the second over $\mathbb{F}_4$ and the third is also constructed over $\mathbb{F}_4$ but with a $0/1$ LT code. 
It can be observed how the upper bounds are very tight. 
Furthermore, as $\absoverhead$ increases the performance of the ensemble with a $0/1$ LT code quickly converges to that of the ordinary ensemble over $\mathbb{F}_4$.

\subsection{Multi-Edge Type Raptor Code Ensembles}

Next we consider multi-edge type Raptor codes with a bivariate LT output degree distribution given by
$\Omega_{\mathsf A} (\x) \left( \z^2 + \z^3\right)/2$.\footnote{This degree distribution is inspired by the one used in RaptorQ codes \cite{lubyraptorq}, where for type A intermediate symbols (called LT symbols in \cite{lubyraptorq}) a conventional LT output degree distribution is used, whereas for type B intermediate symbols (referred to as permanently inactivated symbols in \cite{lubyraptorq}) degrees $2$ and $3$ are chosen with probability $1/2$. See  \cite{lubyraptorq} for more details.}

We consider first a multi-edge type Raptor code over $\GF{2}$ where the outer code is a $(1023,1013)$ Hamming code, with $h_A= 900$ intermediate symbols of type A and $h_B=123$ intermediate symbols of type B.
In order to obtain the bivariate weight enumerator of the Hamming code, the bivariate weight enumerator of the dual code was first obtained by enumerating all its codewords. Then, the extension of the MacWilliams identity developed in Appendix~\ref{sec:appendix_MacWilliams} was applied.
Fig.~\ref{fig:Hamming_MET}  shows the average decoding failure probability, as well as its upper bound. It can be observed how the upper bound is tight.
\begin{figure}[t]
    \centering
    \includegraphics[height=7.75cm]{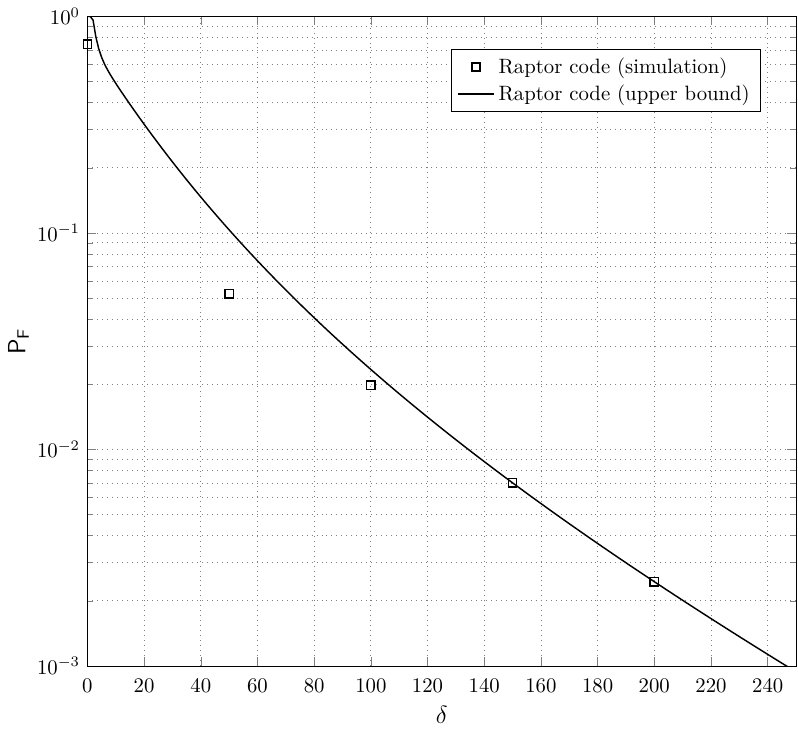}
    \caption{Decoding failure probability $\Pf$ vs absolute overhead for a multi-edge type Raptor code where the outer code is a 
    $(1023,1013)$ Hamming code with $h_A = 900$ and $h_B=123$.
    LT distribution: $\Omega_{\mathsf{A}}(x)(z^2 + z^3)/2$. Line: upper bound. Markers: simulation results.}
    \label{fig:Hamming_MET}
\end{figure}

Next, we consider multi-edge type Raptor code ensembles where the outer code is again drawn from the $(d_v, d_c)$ regular LDPC code ensemble. In particular, 
the outer code is randomly drawn from the $(5,55)$ regular LDPC ensemble with $k=100$ input symbols and $h=110$ intermediate symbols. Out of the $110$ intermediate symbols, $100$ are of class $A$ and $10$ of class $B$. The average bivariate weight enumerator for this ensemble is given by
\begin{equation}
  \weoensemble_{a,b}={\frac{\binom{h_A}{a}\binom{h_B}{b}}{\binom{h}{a+b}}} \weoensemble_{a+b}.
\end{equation}
from which the average bivariate composition enumerator can be obtained through Proposition~\ref{prop:comp_ensemble_bivar} in  Appendix~\ref{sec:appendix_composition}.

Fig.~\ref{fig:LDPC_MET} shows the average probability of decoding failure for three ensembles of multi-edge type Raptor codes, one constructed over  $\GF2$,  another over $\GF4$, and a third one also constructed over $\GF4$ but with a $0/1$ LT code.
It can be observed how the upper bounds  are very tight in this case too. If we compare the the probability of failure of the two ensembles built over $\GF4$, we can see how their performance is almost the same. It is remarkable how restricting the LT code to use only binary labels does not result in an appreciable performance loss.

\section{Code Design Examples} \label{sec:code_design}
In this section we provide several code design examples that illustrate the practical impact of the derived bounds.

\begin{figure}
    \centering
    \includegraphics[height=7.51cm]{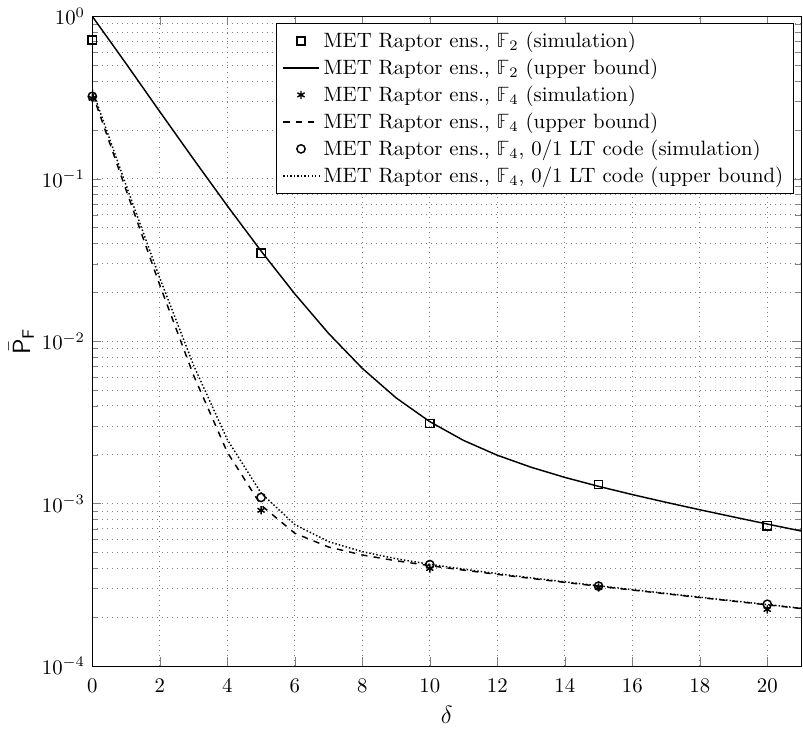}
    \caption{Average probability of decoding failure $\barPf$ vs absolute overhead for three multi-edge type Raptor code ensembles where the outer code is randomly drawn from the $(5,55)$ regular LDPC ensemble  with $k=100$ and $h=110$, with $h_A = 100$ and $h_B=10$.
    LT distribution: $\Omega_{\mathsf{A}}(x)(z^2 + z^3)/2$. Lines: upper bounds. Markers: simulation results.}\label{fig:LDPC_MET}
\end{figure}

\subsection{Design of a Binary Raptor code with an LDPC Outer Code}\label{sec:code_design:inner}

We consider the case in which the outer code ensemble is given and run a computer search in order to find an LT output degree distribution that optimizes a given metric subject to some design constraints. 
In particular, we consider Raptor code ensembles where the outer code is picked from the ${(d_v=3, d_c=33)}$ binary regular LDPC ensemble with $k=1000$  and $h=1100$, and we set as requirement 
minimizing the inactivation decoding complexity subject to a decoding failure probability not exceeding $10^{-3}$.

Inactivation decoding \cite{shokrollahi2005systems} is the efficient \ac{ML} decoding algorithm  
used to decode standardized Raptor codes \cite{MBMS16:raptor,luby2007rfc}.
It can be seen as an extension of iterative (peeling) decoding where, whenever the iterative decoding process stops, an input symbol is declared as inactive, so that iterative decoding is resumed. At the end, one is left with a number of input symbols that have been inactivated, and whose values have to be recovered by means of Gaussian elimination. 
After doing so, all input symbols can be resolved  by back-substitution (i.e., using iterative decoding).
The complexity of inactivation decoding is generally dominated by the Gaussian elimination step, whose complexity is cubic on the number of inactivations. 
Thus, minimizing the number of inactivations can be used as a proxy for minimizing the decoding complexity.

The degree distribution $\Omegarten$, given in \eqref{eq:dist_mbms}, has been designed for inactivation decoding. However, as it can be observed in Fig.~\ref{fig:LDPC_Design}, if we use $\Omegarten$  we do not fulfill the probability of failure constraint, since there is an error floor around $2 \times 10^{-3}$. Thus, we need carry out an ad-hoc design.

The analysis presented in \cite{lazaro2017inactivation} can be used to determine the expected number of inactivations for LT codes. Extending the analysis to Raptor codes is not easy, but, as it was shown in \cite{lazaro2017inactivation}, when the parity-check matrix of the outer code is considerably denser than the generator matrix of the inner LT code, it is possible to design Raptor codes that require few inactivations by optimizing the LT output degree distribution in isolation.\footnote{Note that this heuristic observation holds true also for the case where the outer code parity-check matrix is not dense, e.g., to the case where the outer code is an \ac{LDPC} code, provided that the average check node degree of the \ac{LDPC} code is considerably larger than the average output degree of the \ac{LT} code.} In other words, if we design an LT degree distribution that requires few inactivations, and then construct a Raptor code using this degree distribution for the inner LT code,  we obtain a Raptor code that requires few inactivations.

Following this approach, we can use simulated annealing \cite{kirkpatrick1983optimization} to design an LT degree distribution that minimizes the number of inactivations for the LT code in isolation, under the constraint on the decoding failure probability for the resulting Raptor code, estimated using the upper bounds derived in this paper.
By using this approach we obtained the following degree distribution
\begin{align}
    \Omega_{\mathsf{B}} & = \mbox{ \small $0.0108 \x + 0.4557 \x^2  +  0.1959 \x^3   + 0.1195 \x^4   + 0.0245 \x^5$}  \\[-4mm]
     & \mathrel{\phantom{=}} \mbox{ \small $+0.0243 \x^6 + 0.0357  \x^{10}
    +  0.0412 \x^{11} + 0.0440 \x^{15} $}\\
    & \mathrel{\phantom{=}} \mbox{ \small $+ 0.0196 \x^{21}   + 0.0115 \x^{26} 
    + 0.0088 \x^{30} + 0.0085 \x^{40}$}. \label{deg_dist_ldpc}
\end{align}

\begin{figure}
    \centering
    \includegraphics[width=\columnwidth]{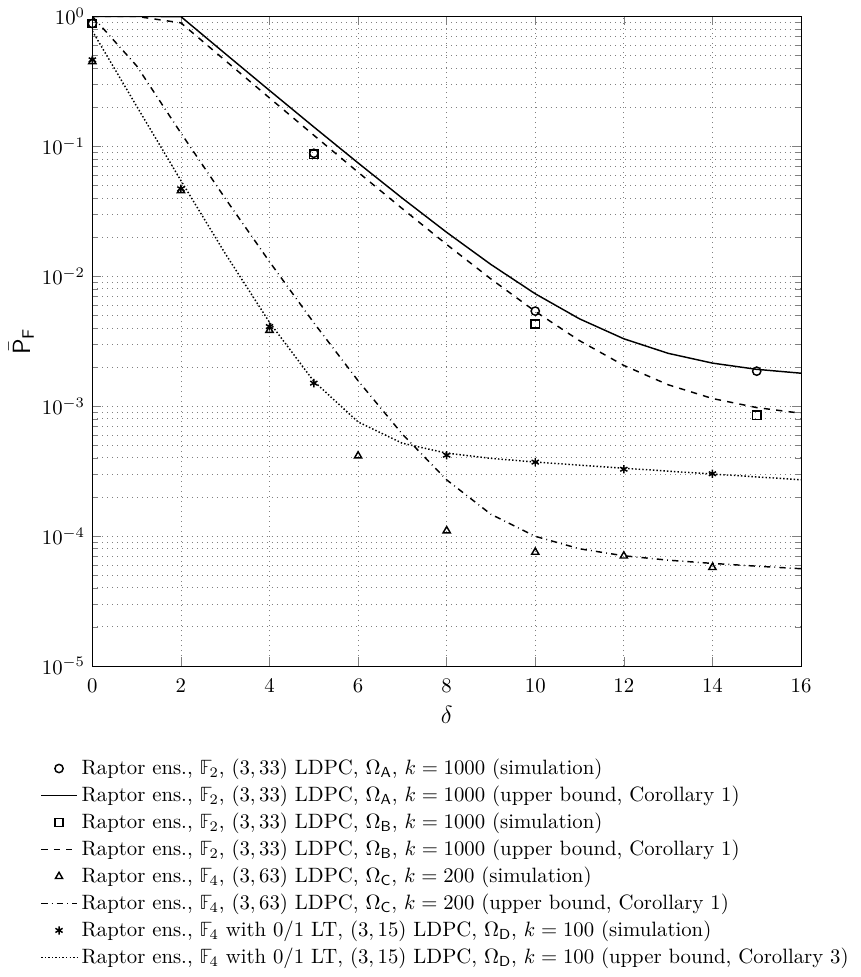}
    \caption{Average probability of decoding failure $\barPf$ vs absolute overhead for 4 different Raptor code ensembles.  
    The first and second ensemble have outer codes randomly drawn from the binary $(d_v=3, d_c= 33)$ regular LDPC ensemble  with $k=1000$ input symbols. For the second and third ensembles the outer codes are randomly drawn from the $(d_v=3, d_c= 63)$ regular LDPC ensemble  with $k=200$ and the $(d_v=3, d_c= 15)$ regular LDPC ensemble  with $k=100$. 
    The LT degree distributions are $\Omega_{\mathsf{A}}$, $\Omega_{\mathsf{B}}$, $\Omega_{\mathsf{C}}$ and $\Omega_{\mathsf{D}}$, respectively.
    Lines: upper bounds. Markers: simulation results.}\label{fig:LDPC_Design}
\end{figure}

Fig.~\ref{fig:LDPC_Design} shows the average probability of decoding failure and its upper bound in Corollary~\ref{corollary:rateless} for the designed ensemble based on $\Omega_{\mathsf{B}}$.  We can observe how the Raptor code ensemble meets the design requirement, since $\barPf < 10^{-3}$ for $\absoverhead=15$. 

If we now consider the number of inactivations, we have that the designed Raptor code ensemble, which employs $\Omega_{\mathsf{B}}$,  needs in average $94$ inactivations for an absolute overhead $\absoverhead=15$. In constrast, the Raptor code ensemble employing $\Omegarten$ needs $87$.
This confirms how a reduction in the number of inactivations forces the failure rate to jump above the maximum tolerable value.

\subsection{Design of a Nonbinary Raptor code with an LDPC Outer Code}\label{sec:code_design:nonbinary}

This design example is similar to the previous one, but this time we focus on a nonbinary Raptor code ensemble. In particular we aim at designing a Raptor code ensemble over $\GF4$, where the outer code is taken from the  $(d_v=3, d_c=63)$ regular LDPC ensemble with $k=200$  and $h=210$. The goal is minimizing the number of inactivations\footnote{The analysis in \cite{lazaro2017inactivation} is also valid for non-binary codes. The number of inactivations is a product of the first phase of inactivation decoding, triangulation, which is equivalent to column and row swapping and does not carry out any operations over the finite field. Thus, the number of inactivations only depends on the elements of the generator matrix of the LT code being zero or nonzero, and not on the particular value in  $\GF{q} \setminus \{0\}$ that the elements take.} subject to $\barPf \leq 10^{-4}$ at $\delta=10$. Using simulated annealing, the following degree distribution is obtained:
\begin{align}
    \Omega_{\mathsf{C}}(\x) &= \mbox{ \small $ 0.0214 \x + 0.3213 \x^2  +  0.2971 \x^3   + 0.0276 \x^4   $} \\
    & \mathrel{\phantom{=}} \mbox{ \small  $ + 0.0252 \x^5  +0.0418 \x^9    + 0.0458  \x^{13}  
    +  0.0654 \x^{18}  $}\\
    & \mathrel{\phantom{=}} \mbox{ \small $ + 0.0457 \x^{23}+ 0.0612 \x^{30} + 0.0295 \x^{35} 
    + 0.0180 \x^{40} $}. \label{deg_dist_nonbinary}
\end{align}
Fig.~\ref{fig:LDPC_Design} shows the average probability of decoding failure for the  ensemble obtained from the code design. We can observe how the constraint on $\barPf$ is fulfilled. The average number of inactivation needed for decoding at $\absoverhead=10$ is approximately~$32$.

\subsection{Design of a Raptor Code with a  $0/1$ LT Code}\label{sec:code_design_met_etc}

We now address the design of a nonbinary Raptor code ensemble with a $0/1$ LT code. We aim at designing a Raptor code ensemble over $\GF4$, where the outer code is taken from the  $(d_v=3, d_c=15)$ regular LDPC ensemble with $k=100$  and $h=125$. The goal is minimizing the number of inactivations subject to $\barPf \leq 2 \times 10^{-3}$ at $\delta=5$. Using simulated annealing, the following degree distribution is obtained:
\begin{align}
    \Omega_{\mathsf{D}}(\x) & = 0.0095 \x + 0.3896 \x^2  +  0.3159 \x^3   + 0.0843 \x^4   \\
    & \mathrel{\phantom{=}} + 0.0611 \x^{10}  +0.0585 \x^{15} + 0.0811  \x^{22}. \label{deg_dist_01}
\end{align}
Fig.~\ref{fig:LDPC_Design} shows the average probability of decoding failure for the designed ensemble. We can observe how the constraint on $\barPf$ is fulfilled. The average number of inactivations needed for decoding at $\absoverhead=5$ is approximately $22$.

\section{Conclusions}\label{sec:Conclusions}

In this paper we have considered different Raptor code constructions over  $\GF{q}$ under \acs{ML} decoding, deriving tight upper and lower bounds to the probability of decoding failure. 
The bounds are first derived for Raptor codes with a deterministic outer code, and then they are extended to Raptor code ensembles in which the outer code is drawn at random from an ensemble of linear block codes.
In all cases the upper bounds require the knowledge of the weight enumerator of the outer code (ensemble) or its composition enumerator, whereas the lower bounds require the knowledge of the joint weight/composition enumerators of the outer code (ensemble).
By means of extensive simulations we have illustrated how the bounds presented in this paper are tight.
A framework for the analysis of the error exponent of Raptor code ensemble sequences is introduced, which allows deriving a lower bound on the error exponent. The result allows gaining further insights on the performance of Raptor code ensemble sequences, by identifying relative overhead regions where an exponential (in the input block size) decay of the error probability can be achieved.
The work is completed by selected examples of Raptor code design based on the bounds derived in this paper.
To the best of the authors' knowledge, this is the first work which considers Raptor codes with a generic $q$-ary outer code. An open question relates to the concentration properties of Raptor code ensembles.

\appendices

\section{Sum of Random Uniform Variables in $\GF{2^m} \backslash \{0\}$}\label{sec:appendix_sum}

The following lemma is used in the proof of Theorem~\ref{theorem:rateless}.

\begin{lemma}\label{lemma:galois}
Let $X_1$, $X_2$ ... $X_l$  be discrete i.i.d random variables uniformly distributed over $\GF{2^m} \backslash \{0\}$. Then
\[
\Pr \{X_1 + X_2+ \hdots + X_l = 0 \}= \frac{1}{q} \left( 1 + \frac{(-1)^l}{(q-1)^{l-1}}\right)
\]
where $q=2^m$.
\end{lemma}

\begin{proof}
Observe that the additive group of $\GF{2^m}$ is isomorphic to the vector space $\mathbb Z_2^m$. Thus, we may let
$X_1$, $X_2$ ... $X_l$  be i.i.d random variables with uniform probability mass function over the vector space $\mathbb Z_2^m \backslash \{0\}$.

Let us introduce the auxiliary random variable $W = X_1 + X_2+ \hdots + X_l$ and let us denote by $P_W(w)$ and by $P_X(x)$ the probability mass functions of $W$ and $X_i$, respectively, where
\[
P_X(x) =
\begin{cases}
 0 & \text{if } x=0 \\
 \frac{1}{q-1} & \text{otherwise.}
 \end{cases}
\]
Due to independence we have $P_W = P_X \ast  P_X \ast \hdots \ast P_X$ which, taking the $m$-dimensional  two-points \ac{DFT} $\msr{J} \{\cdot\}$ of both sides, yields $ \msr J \{ P_W(w) \} =  \left( \msr J \{ P_X(x) \} \right)^l$. Next, since
\[
\hat P_X(t) = \msr J \{ P_X(x) \}=
\begin{cases}
 1 & \text{if } t=0 \\
 \frac{-1}{q-1} & \text{otherwise}
 \end{cases}
\]
we have
\[
\hat P_W(t) = \msr J \{ P_W(w) \} =
\begin{cases}
 1 & \text{if } t=0 \\
 \frac{(-1)^l}{(q-1)^l} & \text{otherwise.}
 \end{cases}
\]
We are interested in  $P_W(0)$ whose expression corresponds to
\begin{equation}
 P_W(0) = \frac{1}{q} \sum_t \hat P_W(t) = \frac{1}{q} + \frac{1}{q} (q-1)  \frac{(-1)^l}{(q-1)^l}
\end{equation}
from which the statement follows.
\end{proof}

The result in this lemma appears in \cite{schotsch:2013}. However, the proof in \cite{schotsch:2013} uses a different approach based on a known result on the number of closed walks of length $l$ in a complete graph of size $q$ from a fixed but arbitrary vertex back to itself.

\section{An Extension of the MacWilliams Identity}\label{sec:appendix_MacWilliams}

Consider a linear block code $\precodegeneric \subset \GF{q}^h$. The same way we defined its bivariate weight enumerator in \eqref{eq:weo_bivariate}, we can define its $h$-variate enumerator polynomial as
\[
\weo(x_1,\hdots, x_h) = \sum_{i_1=0}^{1} \hdots \sum_{i_h=0}^{1}  \weo_{i_1, \hdots, i_h} \prod_{j=1}^{h} x_j^{i_j}
\]
where $ \weo_{i_1, \hdots, i_h}$ denotes the multiplicity of codewords with $\hw(v_1) =i_1$, $\hw(v_2) =i_2$, ... and $\hw(v_h) =i_h$, i.e., the number of codewords with support $(i_i,i_2,\dots,i_h)$.
The following proposition establishes an extension of the MacWilliams identity for $h$-variate weight enumerators.

\begin{prop}\label{proposition:MacWilliams}
Let $\precodegeneric$ be an $(h,k)$ linear block code over $\GF{q}$ with $h$-variate weight enumerator $ \weo(x_1,\hdots, x_h)$. Let $\precodegeneric^\perp$ be the dual of $\precodegeneric$ and denote its $h$-variate weight enumerator by $ \weodual(x_1,\hdots, x_h)$. Then
\begin{align}\label{eq:macwi}
 \weodual(x_1,\hdots, x_h) &= q^{-k}  \prod_{i=1}^{h} \left( 1 + (q-1) x_i \right) \\
 & \mathrel{\phantom{=}} \times 
  \weo \left( \frac{1-x_1}{1+ (q-1) x_1} \, \hdots , \frac{1-x_h}{1+ (q-1) x_h} \right).
\end{align}
\end{prop}
\begin{proof}
The proof builds on that that of the MacWilliams identity for linear block codes over $\mathbb {F}_{q}$ \cite{vanlint:book}.
We start by rewriting $ \weo(x_1,\hdots, x_h)$ as
\[
\weo(x_1,\hdots, x_h) = \sum_{\vecv \in \precodegeneric} \prod_{i=1}^{h} {x_i}^{\hw(v_i)}
\]
Let us now  define function $g(\vecu)$ as follows
\[
g(\vecu) = \sum_{\vecv \in  \mathbb {F}_{q}^h } \charact \left( \langle \vecu , \vecv \rangle \right)  \prod_{i=1}^{h} {x_i}^{\hw(v_i)}
\]
where $\charact$ is a non-trivial character of $\left( \mathbb {F}_{q} , + \right)$.

We have
\begingroup
\allowdisplaybreaks
\begin{align}\label{eq:mw1}
  \sum_{ \vecu \in \precodegeneric } g(\vecu) &=  \sum_{ \vecu \in \precodegeneric } \sum_{\vecv \in  \mathbb {F}_{q}^h } \charact \left( \langle \vecu , \vecv \rangle \right)
  \prod_{i=1}^{h} {x_i}^{\hw(v_i)} \\ 
  & =  \sum_{\vecv \in  \mathbb {F}_{q}^h }    \prod_{i=1}^{h} {x_i}^{\hw(v_i)}  \sum_{ \vecu \in \precodegeneric }  \charact \left( \langle \vecu , \vecv \rangle \right) \\
  &= \sum_{\vecv \in  \precodegeneric^\perp }   \prod_{i=1}^{h} {x_i}^{\hw(v_i)}  \sum_{ \vecu \in \precodegeneric }  \charact \left( \langle \vecu , \vecv \rangle \right) \\
  & \mathrel{\phantom{=}} 
  +  \sum_{\vecv \notin  \precodegeneric^\perp }    \prod_{i=1}^{h} {x_i}^{\hw(v_i)}  \sum_{ \vecu \in \precodegeneric }  \charact \left( \langle \vecu , \vecv \rangle \right)  \\
  &=  \sum_{\vecv \in  \precodegeneric^\perp }  \prod_{i=1}^{h} {x_i}^{\hw(v_i)}   \sum_{ \vecu \in \precodegeneric }  \charact \left( 0 \right) = \sum_{\vecv \in  \precodegeneric^\perp }   \prod_{i=1}^{h} {x_i}^{\hw(v_i)}   \, | \precodegeneric |  \\
  &= | \precodegeneric | \,  \weodual(x_1,\hdots, x_n)
\end{align}
\endgroup

Let us now rewrite $g(\vecu)$ as follows
\begin{align}\label{eq:gu}
  g(\vecu) &=  \sum_{\vecv \in  \mathbb {F}_{q}^h }   \prod_{i=1}^{h} {x_i}^{\hw(v_i)}   \, \charact \left(  u_1 v_1 + \hdots + u_h v_h \right) \\
   &=  \sum_{\vecv \in  \mathbb {F}_{q}^h }  \prod_{i=1}^{h} {x_i}^{\wh(v_i)}  \charact \left(  u_i v_i \right) \\
   &=  \prod_{i=1}^{h} \sum_{v \in  \mathbb {F}_{q} }  {x_i}^{\wh(v)}  \charact \left(  u_i v \right)
\end{align}

Let us now look at the inner summation, we have

\[
\sum_{v \in \mathbb {F}_{q} } \mkern-4mu {x_i}^{\wh(v)}  \charact \left(  u_i v \right) =
\begin{cases}
1+ (q-1) x_i, & \mbox{if } u_i=0 \\
1 + x_i \mkern-22mu \sum \limits_{\alpha \in \mathbb {F}_{q} \backslash \{ 0\}} \mkern-20mu \charact \left(  \alpha \right)= 1 - x,\mkern-10mu  & \mbox{otherwise.}
\end{cases}
\]

Thus, we can write
\begin{align}\label{eq:gu2}
  g(\vecu) &= \prod_{i=1}^{h} (1-x_i)^{\wh(v_i)} \left( 1 + (q-1) x_i \right)^{1 - \wh(v_i)}
\end{align}

Finally, if we replace \eqref{eq:gu2} into \eqref{eq:mw1} we obtain
\begingroup
\allowdisplaybreaks
\begin{align}\label{eq:bxz}
 \weodual(x,z) &= \frac{1}{| \precodegeneric |} \sum_{ \vecu \in \precodegeneric } g(\vecu)  \\
 &= \frac{1}{| \precodegeneric |}   \sum_{ \vecu \in \precodegeneric } \prod_{i=1}^{h} (1-x_i)^{\wh(v_i)} \left( 1 + (q-1) x_i \right)^{1 - \wh(v_i)}  \\
   &= q^{-k}  \prod_{i=1}^{h} \left( 1 + (q-1) x_i \right) \\
   & \qquad \qquad \, \,\times  \weo \left( \frac{1-x_1}{1+ (q-1) x_1} \, {\hdots} , \frac{1-x_h}{1+ (q-1) x_h} \right)
\end{align}
\endgroup
\end{proof}
The result in Proposition~\ref{proposition:MacWilliams} is strongly related to the result derived in \cite[Appendix]{battail1979replication}, where a similar analysis is used to derive a maximum-a-posteriori decoding algorithm for a code based on its dual. However, for the sake of completeness, we decided to include the result in the form of a Theorem with its corresponding proof.

Now that we have a MacWilliams identity for $h$-variate weight enumerators it is easy to derive a similar result for bi-variate weight enumerators.
\begin{prop}\label{theorem:MacWilliams_bivariate}
Let $\precodegeneric$ be an $(h,k)$ linear block code over $\GF{q}$ in which the $h$ codeword symbols are divided into $h_A$ symbols of class $A$ and $h_B=h-h_A$ of class $B$, with bivariate weight enumerator of $\weo(x,z)$. Let $\precodegeneric^\perp$ be the dual of $\precodegeneric$ and denote its bivariate weight enumerator by $\weodual(x,z)$. Then
\begin{align}\label{eq:macwibi}
\weodual(x,z) &= q^{-k} \left( 1 + (q-1) x \right)^{h_A} \left( 1 + (q-1) z \right)^{h_B} \\
& \mathrel{\phantom{=}} \times
\weo \left( \frac{1-x}{1+ (q-1) x} \, , \frac{1-z}{1+ (q-1) z} \right).
\end{align}
\end{prop}

\begin{proof}
We just need to introduce the variable changes $x_i=x$ for $i=1,\hdots, h_A$ and  $x_i=z$ for $i=h_A+1,\hdots, h$ in  Proposition~\ref{proposition:MacWilliams}.
\end{proof}
Note that the special case of Proposition~\ref{theorem:MacWilliams_bivariate} for $h_B=h_A$ is proposed in \cite[Chapter~5.6]{MacWillimas77:Book} as an exercise.

\section{Average Composition Enumerators of Some Codes Ensembles}\label{sec:appendix_composition}

This appendix provides results on the average composition enumerator of some code ensembles. The following proposition states that, in some cases, the average composition enumerator can be easily derived from the average weight enumerator.
\begin{prop}\label{prop:comp_ensemble}
Consider an ensemble $\codensemble$ of linear block codes, all with block length $h$, along with a probability measure on each such code. Let $\weoensemble_l$ be the expected weight enumerator of a random code $\precodegeneric \in \codensemble$. Assume that $\Pr\{ \vecv \in \precodegeneric | \comp(\vecv)= \labv \} = \Pr\left\{ \vecv \in \precodegeneric | \wh(\vecv) = \sum_{i=1}^{q-1} \lab_i \right\}$ for all $\vecv \in \GF{q}^h$. Then
\begin{equation}\label{eq:composition_from_weight_enum}
\lenumensemble  = \weoensemble_l \binom{l}{\lab_1, \lab_2, \hdots, \lab_{q-1}} \left(q-1\right) ^{-l} 
\end{equation}
where $\l = \sum_{i=1}^{q-1} \lab_i$.
\end{prop}
\begin{IEEEproof}
We can express $\lenumensemble$ as the number of vectors of composition $\labv$ times the probability that each such vector is a codeword. Letting ${l =\sum_{i=1}^{q-1} \lab_i = \wh(\vecv)}$ we can write
\begingroup
\allowdisplaybreaks
\begin{align*}
\lenumensemble &= \binom{h}{\labv} \Pr\{ \vecv \in \precodegeneric | \comp(\vecv)= \labv \} = \binom{h}{\labv} \Pr \{ \vecv \in \precodegeneric | \wh(\vecv) = l \} \notag \\
&= \binom{h}{\labv}  \frac{\weoensemble_{\l}}{{h \choose l}(q-1)^{\l}} \end{align*}
\endgroup
The last obtained expression yields \eqref{eq:composition_from_weight_enum} by applying the identity $\binom{h}{\labv}=\binom{h}{l} \binom{l}{f_1,f_2,\dots,f_{q-1}}$.
\end{IEEEproof}
Examples of ensembles for which the assumption on Proposition~\ref{prop:comp_ensemble} holds are the uniform parity-check ensemble and the (regular and irregular) \ac{LDPC} code ensembles.
\subsubsection{Uniform parity-check ensemble} For a uniform parity-check ensemble defined by a random parity-check matrix of size $(h-k) \times h$ with \ac{i.i.d.} entries uniformly distributed in $\GF{q}$ we have $\weoensemble_{\l} = \binom{h}{\l} (q-1)^{\l} q^{-(h-k)}$ and therefore \eqref{eq:composition_from_weight_enum} leads to
\[
\lenumensemble = \binom{h}{\labv} \, q^{-(h-k)}.
\]
\subsubsection{Regular \ac{LDPC} ensemble} Consider a $(d_v, d_c)$ regular \ac{LDPC} code ensemble of length $h$, where $d_v$ and $d_c$ are the variable and check node degrees, respectively. The ensemble is defined by all possible permutations of the $h d_v = (h-k) d_c$ edges between check and variable node sockets and by all possible ways to label the edges with nonzero symbols. Each edge permutation is picked with uniform probability and the label of each edge is drawn uniformly at random in $\GF{q} \backslash \{ 0\}$. The average weight enumerator for this ensemble is given by \cite{burshtein:04,kasai:2008}
\begin{align}\label{eq:ldpc}
  \weoensemble_l &= \binom{h}{l}
  \frac{\text{coeff} \left( p(x)^{h \, d_v /d_c}, x^{l \, d_v} \right) }{\binom{h \, d_v}{ l \, d_v} (q-1)^{ l ( d_v -1) }}
\end{align}
where
$p(x) = \frac{1}{q} \left( 1 + (q-1) x \right)^{d_c} + \frac{q-1}{q} (1-x)^{d_c}$. Hence, applying \eqref{eq:composition_from_weight_enum} we obtain
\begin{align*}
\lenumensemble = \binom{h}{\labv} \frac{\text{coeff} \left( p(x)^{h \, d_v /d_c}, x^{l d_v} \right)}{\binom{h \, d_v}{ l \, d_v}} (q-1)^{- \l d_v}
\end{align*}
Proposition~\ref{prop:comp_ensemble} can be extended to bivariate enumerators using the same proof argument.
\begin{prop}\label{prop:comp_ensemble_bivar}
Consider an ensemble $\codensemble$ of linear block codes, all with block length $h=h_A+h_B$, along with a probability measure on each such code. Let $\weoensemble_{l,s}$ be the expected bivariate weight enumerator of a random code $\precodegeneric \in \codensemble$. Assume that $\Pr\{ \vecv \in \precodegeneric | \comp(\vecv_A)= \labv_A, \comp(\vecv_B)= \labv_B \} = \Pr\left\{ \vecv \in \precodegeneric | \wh(\vecv_A) = \sum_{i=1}^{q-1} \lab_{A,i}, \wh(\vecv_B) = \sum_{i=1}^{q-1} \lab_{B,i} \right\}$ for all $\vecv = (\vecv_A, \vecv_B) \in \GF{q}^h$. Then
\begin{align}
\lenumbiensemble &= \weoensemble_{l,s}  \, \binom{l}{\lab_{A,1}, \lab_{A,2}, \hdots, \lab_{A,q-1}} \left(q-1\right) ^{-l} \\ 
& \mathrel{\phantom{=}} \times 
\binom{s}{\lab_{B,1}, \lab_{B,2}, \hdots, \lab_{B,q-1}} \left(q-1\right) ^{-s}
\label{eq:composition_from_weight_enum_bi}
\end{align}
where $\sum_{i=1}^{q-1} l = \lab_{A,i}$ and $s = \sum_{i=1}^{q-1}$.
\end{prop}

\section{Average Bicomposition Enumerator of Uniform Parity-Check Ensembles}\label{app:bicomp}
This appendix provides results on the average bicomposition and biweight enumerators of some ensembles.
\begin{prop}\label{prop:rand_bicomp}
Consider the uniform parity-check ensemble defined by a random parity-check matrix of size $(h-k) \times h$ with \ac{i.i.d.} entries with uniform distribution in $\mathbb F_q$. For all $\bkappa \in \cK_{q,h}$, the expected joint composition enumerator for a random code drawn for the ensemble is
\begin{align*}
\avgJCEF_{\bkappa} = {h \choose \bkappa} q^{-2(h-k)}.
\end{align*}
\end{prop}
\begin{IEEEproof}
The parameter $\avgJCEF_{\bkappa}$ may be expressed as the total number of pairs $(\mathbf{r}_1,\mathbf{r}_2) \in \bbF{q}^h \times \bbF{q}^h $ with joint composition $\bkappa$, times the probability that both $\mathbf{r}_1$ and $\mathbf{r}_2$ are codewords given that their joint composition is $\bkappa$. Hence, we can write
 \begin{align*}
\avgJCEF_{\bkappa} &= {h \choose \bkappa} \Pr\{ \{\mathbf{r}_1 \transpose{\mathbf{H}} = \mathbf{0}\} \cap \{ \mathbf{r}_2 \transpose{\mathbf{H}} = \mathbf{0}\} | \jc(\mathbf{r}_1,\mathbf{r}_2) = \bkappa \}  \\
& = {h \choose \bkappa} \mathsf{p}_{\bkappa}^{h-k}
\end{align*}
where, letting $\mathbf{h}$ be the generic row of $\mathbf{H}$, $\mathsf{p}_{\bkappa} = \Pr \{ \{\mathbf{r}_1 \transpose{\mathbf{h}} = 0 \} \cap \{ \mathbf{r}_2 \transpose{\mathbf{h}} = 0\} | \jc(\mathbf{r}_1,\mathbf{r}_2) = \bkappa \}$. If $\bkappa \in \cK_{q,h}$ then five different cases may occur; next we show that in all of them we have $\mathsf{p}_{\bkappa} = q^{-2}$. We repeatedly exploit the following property: if $\mathbf{r} \in \bbF{q}^h$ and $\mathbf{h}$ is a random vector in $\bbF{q}^h$ whose elements are uniform i.i.d. random variables in $\bbF{q}$, then $\Pr \{ \mathbf{r} \, \transpose{\mathbf{h}} = \beta \} = q^{-1}$ for all $\beta \in \bbF{q}$. For the sake of notational simplicity, we denote by $E_{\bkappa}$ the event that $\jc(\mathbf{r}_1,\mathbf{r}_2) = \bkappa$.\\
\indent\emph{Case 1:} $|\bkappa_1|>0$, $|\bkappa_2|>0$, $|\bkappa_3|>0$ ($\mathbf{r}_1$ and $\mathbf{r}_2$ have partially overlapping supports). Without loss of generality, assume $\mathbf{r}_1 = (\mathbf{r}_{1,1} | \mathbf{r}_{1,2} | \mathbf{0} | \mathbf{0})$ and $\mathbf{r}_2=(\mathbf{0} | \mathbf{r}_{2,1} | \mathbf{r}_{2,2} | \mathbf{0})$, where $\mathbf{r}_{1,1}$, $\mathbf{r}_{1,2}$, $\mathbf{r}_{2,1}$, and $\mathbf{r}_{2,2}$ are nonzero and all subvectors occupying the same position have the same length. Letting $\mathbf{h}=(\mathbf{h}_1 | \mathbf{h}_2 | \mathbf{h}_3 | \mathbf{h}_4)$ we have $\mathsf{p}_{\bkappa} = \Pr \{ {\{\mathbf{r}_{1,1} \transpose{\mathbf{h}}_1 + \mathbf{r}_{1,2} \transpose{\mathbf{h}}_2 = 0 \}} \cap \{\mathbf{r}_{2,1} \transpose{\mathbf{h}}_2 + \mathbf{r}_{2,2} \transpose{\mathbf{h}}_3 = 0 \} | E_{\bkappa} \} = \Pr \{ \mathbf{r}_{1,1} \transpose{\mathbf{h}}_1 + \mathbf{r}_{1,2} \transpose{\mathbf{h}}_2 = 0 | E_{\bkappa} \} \, \Pr \{ \mathbf{r}_{2,1} \transpose{\mathbf{h}}_2 + \mathbf{r}_{2,2} \transpose{\mathbf{h}}_3 = 0 | E_{\bkappa} \} = (q^{-1}) (q^{-1}) = q^{-2}$, where we exploited independence of $\mathbf{h}_1$, $\mathbf{h}_2$, and $\mathbf{h}_3$.\\
\indent\emph{Case 2:} $|\bkappa_1|>0$, $|\bkappa_2|=0$, $|\bkappa_3|>0$ (the support of $\mathbf{r}_2$ includes that of $\mathbf{r}_1$). Same argument with $\mathbf{r}_{1,1}=\mathbf{0}$.\\
\indent\emph{Case 3:} $|\bkappa_1|=0$, $|\bkappa_2|>0$, $|\bkappa_3|>0$ (the support of $\mathbf{r}_1$ includes that of $\mathbf{r}_2$). Same argument with $\mathbf{r}_{2,2}=\mathbf{0}$.\\
\indent\emph{Case 4:} $|\bkappa_1|>0$, $|\bkappa_2|>0$, $|\bkappa_3|=0$ ($\mathbf{r}_1$ and $\mathbf{r}_2$ have disjoint supports). Same argument with $\mathbf{r}_{1,2}=\mathbf{r}_{2,1}=\mathbf{0}$.\\
\indent\emph{Case 5:} $|\bkappa_1| = |\bkappa_2| = 0$, $|\bkappa_3|>0$, $\kappa_{0,0}+\sum_{i=1}^{q-1} \kappa_{i,(i+b)\mathrm{mod}q}<h$ for all $b \in \{0,\dots,q-2\}$ ($\mathbf{r}_1$ and $\mathbf{r}_2$ have the same support but are not linearly dependent). Let $\mathbf{r}_1=(r_{1,0}, \dots, r_{1,h-1})$, $\mathbf{r}_2=(r_{2,0},\dots,r_{2,h-1})$ and $\mathbf{h}=(\mathrm{h}_0,\dots,\mathrm{h}_{h-1})$. Since $\mathbf{r}_1$ and $\mathbf{r}_2$ are nonzero and not linearly dependent, there exist $s$ and $t$ such that the vectors $(r_{1,s}, r_{1,t})$ and $(r_{2,s}, r_{2,t})$ are linearly independent. Letting $\beta_1=-\sum_{i=0,i\neq s,t}^{h-1} r_{1,i} \mathrm{h}_i$ and $\beta_2=-\sum_{i=0,i\neq s,t}^{h-1} r_{2,i} \mathrm{h}_i$ we obtain $\mathsf{p}_{\bkappa} = \Pr \{ \{ r_{1,s} \mathrm{h}_s + r_{1,t} \mathrm{h}_t = \beta_1 \} \cap \{ r_{2,s} \mathrm{h}_s + r_{2,t} \mathrm{h}_t = \beta_2 \} | E_{\bkappa} \}$. Linear independence of $(r_{1,s}, r_{1,t})$ and $(r_{2,s}, r_{2,t})$ implies that for any $\beta_1$ and $\beta_2$ there exists a unique pair $(\mathrm{h}_s,\mathrm{h}_t)$ fulfilling the two equations. Since all pairs are equiprobable and their number is $q^2$ we have $\mathsf{p}_{\bkappa}=q^{-2}$.
\end{IEEEproof}

The following result is a direct consequence of Proposition~\ref{prop:rand_bicomp} in the binary case.
\begin{prop}\label{lemma:LRC}
Consider the uniform parity-check ensemble defined by a random parity-check matrix of size $(h-k) \times h$ with \ac{i.i.d.} entries with uniform distribution in $\mathbb F_2$. For all $\btau \in \cT_{2,h}$, the expected joint composition enumerator for a random code drawn for the ensemble is
\begin{align*}
\avgJWEF_{\btau} = {h \choose \btau} 4^{-(h-k)}.
\end{align*}
\end{prop}
\begin{IEEEproof}
Recall from Remark~\ref{remark:binary_equiv} that for $q=2$ the two concepts of joint composition and joint weight become equivalent so that, letting $\btau = \jw(\bkappa)$, we can write $\avgJWEF_{\btau} = \avgJCEF_{\bkappa}$. 
\end{IEEEproof}

\section*{Acknowledgment}
The authors would like to thank the Associate Editor and the  anonymous reviewers for their valuable comments, which helped to substantially improve the paper.



\begin{thebibliography}{10}
\providecommand{\url}[1]{#1}
\csname url@samestyle\endcsname
\providecommand{\newblock}{\relax}
\providecommand{\bibinfo}[2]{#2}
\providecommand{\BIBentrySTDinterwordspacing}{\spaceskip=0pt\relax}
\providecommand{\BIBentryALTinterwordstretchfactor}{4}
\providecommand{\BIBentryALTinterwordspacing}{\spaceskip=\fontdimen2\font plus
\BIBentryALTinterwordstretchfactor\fontdimen3\font minus
  \fontdimen4\font\relax}
\providecommand{\BIBforeignlanguage}[2]{{%
\expandafter\ifx\csname l@#1\endcsname\relax
\typeout{** WARNING: IEEEtran.bst: No hyphenation pattern has been}%
\typeout{** loaded for the language `#1'. Using the pattern for}%
\typeout{** the default language instead.}%
\else
\language=\csname l@#1\endcsname
\fi
#2}}
\providecommand{\BIBdecl}{\relax}
\BIBdecl

\bibitem{lazaro:Globecom2016}
F.~L{\'a}zaro, G.~Liva, E.~Paolini, and G.~Bauch, ``Bounds on the error
  probability of {Raptor} codes,'' in \emph{Proc. IEEE Global Commun. Conf.},
  Washington DC, USA, Dec. 2016.

\bibitem{byers02:fountain}
J.~Byers, M.~Luby, and M.~Mitzenmacher, ``A digital fountain approach to
  reliable distribution of bulk data,'' \emph{{IEEE} J. Select. Areas Commun.},
  vol.~20, no.~8, pp. 1528--1540, Oct. 2002.

\bibitem{luby02:LT}
M.~Luby, ``{LT} codes,'' in \emph{Proc. 43rd Annual IEEE Symp. on Foundations
  of Computer Science}, Vancouver, Canada, Nov. 2002, pp. 271--282.

\bibitem{shokrollahi06:raptor}
M.~Shokrollahi, ``Raptor codes,'' \emph{{IEEE} Trans. Inf. Theory}, vol.~52,
  no.~6, pp. 2551--2567, Jun. 2006.

\bibitem{MBMS16:raptor}
{ETSI TS 26.346 V13.3.0}, ``{UMTS; LTE; Multimedia Broadcast / Multicast
  Service; Protocols and Codecs},'' Jan. 2016.

\bibitem{luby2007rfc}
M.~Luby, A.~Shokrollahi, M.~Watson, and T.~Stockhammer, ``{RFC} 5053: Raptor
  forward error correction scheme: Scheme for object delivery,'' {IETF}, Tech.
  Rep., Oct. 2007.

\bibitem{berlekamp1968algebraic}
E.~Berlekamp, \emph{Algebraic coding theory}.\hskip 1em plus 0.5em minus
  0.4em\relax New York: McGraw-Hill, 1968.

\bibitem{lamacchia91:solving}
B.~A. LaMacchia and A.~M. Odlyzko, ``Solving large sparse linear systems over
  finite fields,'' \emph{Advances in Cryptology-CRYPT0’90}, pp. 109--133,
  1991.

\bibitem{fekri:ldpc}
H.~Pishro-Nik and F.~Fekri, ``On decoding of low-density parity-check codes
  over the binary erasure channel,'' \emph{{IEEE} Trans. Commun.}, vol.~50,
  no.~3, pp. 439--454, Mar. 2004.

\bibitem{miller04:bec}
D.~Burshtein and G.~Miller, ``An efficient maximum likelihood decoding of
  {LDPC} codes over the binary erasure channel,'' \emph{{IEEE} Trans. Inf.
  Theory}, vol.~50, no.~11, pp. 2837--2844, Nov. 2004.

\bibitem{shokrollahi2005systems}
M.~Shokrollahi, S.~Lassen, and R.~Karp, ``Systems and processes for decoding
  chain reaction codes through inactivation,'' Feb. 2005, {US} Patent
  6,856,263.

\bibitem{mahdaviani2012raptor}
K.~Mahdaviani, M.~Ardakani, and C.~Tellambura, ``{On Raptor code design for
  inactivation decoding},'' \emph{{IEEE} Commun. Lett.}, vol.~60, no.~9, pp.
  2377--2381, Sep. 2012.

\bibitem{lazaro:ITW}
F.~L{\'a}zaro~Blasco, G.~Liva, and G.~Bauch, ``{LT} code design for
  inactivation decoding,'' in \emph{Proc. 2014 IEEE Inf. Theory Workshop},
  Hobart, Tasmania, Australia, Nov. 2014, pp. 441--445.

\bibitem{lazaro:scc2015}
------, ``Enhancing the {LT} component of {Raptor} codes,'' in \emph{Proc. of
  the 10th Int. ITG Conf. Systems, Commun. and Coding}, Hamburg, Germany, Feb.
  2015.

\bibitem{lazaro2017inactivation}
F.~L{\'a}zaro, G.~Liva, and G.~Bauch, ``Inactivation decoding of {LT} and
  {Raptor} codes: Analysis and code design,'' \emph{{IEEE} Trans. Commun.},
  vol.~65, no.~10, pp. 4114--4127, Oct. 2017.

\bibitem{Rahnavard:07}
N.~Rahnavard, B.~Vellambi, and F.~Fekri, ``Rateless codes with unequal error
  protection property,'' \emph{{IEEE} Trans. Inf. Theory}, vol.~53, no.~4, pp.
  1521--1532, Apr. 2007.

\bibitem{schotsch:2013}
B.~Schotsch, G.~Garrammone, and P.~Vary, ``Analysis of {LT} codes over finite
  fields under optimal erasure decoding,'' \emph{{IEEE} Commun. Lett.},
  vol.~17, no.~9, pp. 1826--1829, Sep. 2013.

\bibitem{Schotsch:14}
B.~E. Schotsch, ``Rateless coding in the finite length regime,'' Ph.D.
  dissertation, Inst. of Commun. Systems and Data Proc., RWTH Aachen, Aachen,
  Germany, Jul. 2014.

\bibitem{Liva10:fountain}
G.~Liva, E.~Paolini, and M.~Chiani, ``{Performance versus overhead for fountain
  codes over $\mathbb{F}_q$},'' \emph{{IEEE} Commun. Lett.}, vol.~14, no.~2,
  pp. 178--180, Feb. 2010.

\bibitem{wang:2015}
P.~Wang, G.~Mao, Z.~Lin, M.~Ding, W.~Liang, X.~Ge, and Z.~Lin, ``Performance
  analysis of {R}aptor codes under maximum likelihood decoding,'' \emph{{IEEE}
  Trans. Commun.}, vol.~64, no.~3, pp. 906--917, Mar. 2016.

\bibitem{lazaro:ISIT2015}
F.~L\'azaro~Blasco, E.~Paolini, G.~Liva, and G.~Bauch, ``On the weight
  distribution of fixed-rate {Raptor} codes,'' in \emph{{Proc. 2015 IEEE Int.
  Symp. Inf. Theory}}, Hong Kong, China, Jun. 2015, pp. 2880--2884.

\bibitem{lazaro:JSAC}
F.~L{\'a}zaro, E.~Paolini, G.~Liva, and G.~Bauch, ``Distance spectrum of
  fixed-rate {Raptor} codes with linear random precoders,'' \emph{{IEEE} J.
  Select. Areas Commun.}, vol.~34, no.~2, pp. 422--436, Feb. 2016.

\bibitem{Zhang:bounds}
K.~Zhang, Q.~Zhang, and J.~Jiao, ``Bounds on the reliability of {RaptorQ} codes
  in the finite-length regime,'' \emph{IEEE Access}, vol.~5, no.~5, pp.
  24\,766--24\,774, Oct. 2017.

\bibitem{lubyraptorq}
{RFC 6330}, ``{Network working group; Request for Comments: 5053; RaptorQ
  Forward Error Correction Scheme for Object Delivery},'' Aug. 2011.

\bibitem{dawson67:inequality}
D.~A. Dawson and D.~Sankoff, ``An inequality for probabilities,''
  \emph{Proc. American Math. Society}, vol.~18, no.~3, pp. 504--507, Jun. 1967.

\bibitem{Barak07}
O.~{Barak} and D.~{Burshtein}, ``{Lower bounds on the error rate of LDPC
  code ensembles},'' \emph{{IEEE} Trans. Inf. Theory}, vol.~53, no.~11, pp.
  4225--4236, Nov 2007.

\bibitem{Bonferroni36}
C.~Bonferroni, ``Teoria statistica classi e calcolo delle
  probabilit{\`a},'' \emph{Pubbl. R. Ist. Super. Sci. Econ. Comm. Firenze},
  vol.~8, pp. 3--62, 1936.

\bibitem{kwerel75:most}
S.~M. Kwerel, ``Most stringent bounds on aggregated probabilities of
  partially specified dependent probability systems,'' \emph{J. Amer. Statist.
  Assoc.}, vol.~70, no. 350, pp. 472--479, Jun. 1975.

\bibitem{MacWillimas77:Book}
F.~{Mac Williams} and N.~Sloane, \emph{The theory of error-correcting
  codes}.\hskip 1em plus 0.5em minus 0.4em\relax North Holland Mathematical
  Libray, 1977, vol.~16.

\bibitem{MMS:1972}
F.~MacWilliams, C.~Mallows, and N.~Sloane, ``Generalizations of {G}leason's
  theorem on weight enumerators of self-dual codes,'' \emph{{IEEE} Trans. Inf.
  Theory}, vol.~18, no.~6, pp. 794--805, Nov. 1972.

\bibitem{shokrollahi2003systematic}
M.~Shokrollahi and M.~Luby, ``Systematic encoding and decoding of chain
  reaction codes,'' Jun. 2005, {US} Patent 6,909,383.

\bibitem{Lazaro:phd}
{F. L{\'a}zaro}, ``Fountain codes under maximum likelihood decoding,'' Ph.D.
  dissertation, Institute for Telecommunications, Hamburg University of
  Technology, Hamburg, Germany, 2017.

\bibitem{luby2011rfc}
M.~Luby, A.~Shokrollahi, M.~Watson, T.~Stockhammer, and L.~Minder, ``{RFC}
  6330: {RaptorQ} forward error correction scheme for object delivery,''
  {IETF}, Tech. Rep., Aug. 2011.

\bibitem{shokrollahi2011raptor}
A.~Shokrollahi and M.~Luby, ``Raptor codes,'' \emph{{Foundations and Trends in
  Commun. and Inf. Theory}}, vol.~6, no. 3-4, pp. 213--322, 2011.

\bibitem{kirkpatrick1983optimization}
S.~Kirkpatrick, D.~Gelatt, and M.~Vecchi, ``Optimization by simmulated
  annealing,'' \emph{Science}, vol. 220, no. 4598, pp. 671--680, 1983.

\bibitem{vanlint:book}
J.~van Lint, \emph{Introduction to Coding Theory}, ser. Graduate Texts in
  Mathematics.\hskip 1em plus 0.5em minus 0.4em\relax Springer Berlin
  Heidelberg, 1998.

\bibitem{battail1979replication}
G.~Battail, M.~Decouvelaere, and P.~Godlewski, ``Replication decoding,''
  \emph{{IEEE} Trans. Inf. Theory}, vol.~25, no.~3, pp. 332--345, May 1979.

\bibitem{burshtein:04}
D.~Burshtein and G.~Miller, ``Asymptotic enumeration methods for analyzing
  {LDPC} codes,'' \emph{{IEEE} Trans. Inf. Theory}, vol.~50, no.~6, pp.
  1115--1131, Jun. 2004.

\bibitem{kasai:2008}
K.~Kasai, C.~Poulliat, D.~Declercq, T.~Shibuya, and K.~Sakaniwa, ``Weight
  distribution of non-binary {LDPC} codes,'' in \emph{Proc. 2008 Int. Symp.
  Inf. Theory and App}, Dec. 2008, pp. 1--6.

\end{thebibliography}
\end{document}